\documentclass[11pt]{article}
\usepackage[letterpaper,top=2cm,bottom=2cm,left=3cm,right=3cm,marginparwidth=1.75cm]{geometry}
\usepackage[utf8]{inputenc}
\usepackage[T1]{fontenc}
\usepackage{amsmath}
\usepackage{amssymb}
\usepackage{amsthm}
\usepackage{amsfonts}
\usepackage{amsxtra}
\usepackage{amscd}
\usepackage{stmaryrd}
\usepackage{mathrsfs}
\usepackage{mathtools}
\usepackage[mathscr]{eucal}
\usepackage[all]{xy}
\usepackage{color}
\usepackage{authblk}
\usepackage{enumerate}
\usepackage{tikz}
\usepackage{tikz-cd}
\usepackage{indentfirst}
\usepackage{bm}
\usepackage{fancyhdr}
\usepackage{latexsym}
\usepackage{graphicx}
\usepackage[export]{adjustbox}
\usepackage[colorlinks=true, allcolors=blue]{hyperref}
\usepackage{cases}
\usepackage{float}
\usepackage{caption}
\usepackage{booktabs}

\usepackage{palatino}
\usepackage{authblk}

\newtheorem{theorem}{Theorem}[section]
\newtheorem{lemma}[theorem]{Lemma}
\newtheorem{corollary}[theorem]{Corollary}
\newtheorem{proposition}[theorem]{Proposition}

\theoremstyle{definition}
\newtheorem{definition}[theorem]{Definition}
\newtheorem{example}[theorem]{Example}

\theoremstyle{remark}
\newtheorem{remark}[theorem]{Remark}

\numberwithin{equation}{section}
\everymath{\displaystyle}

\title{On the KAK Decomposition and Equivalence Classes}

\author[1,2,3]{Dawei Ding}
\author[1]{Yu Liu}
\author[1]{Zi-Wen Liu}

\affil[1]{Yau Mathematical Sciences Center, Tsinghua University, Beijing 100084, China}
\affil[2]{Center for Mathematics and Interdisciplinary Sciences, Fudan University, Shanghai, 200433, China}
\affil[3]{Shanghai Institute for Mathematics and Interdisciplinary Sciences, Shanghai, 200433, China}
\date{\today}

\begin{document}

\maketitle

\begin{abstract}
The KAK decomposition is a fundamental tool in Lie theory and quantum computing. Despite its widespread use, the mathematical foundations remain incomplete, particularly regarding the precise conditions for the decomposition and the characterization of equivalence classes under multiplication by elements of $K$. Here, we present a mathematical theory of the KAK decomposition for connected compact semisimple Lie groups and derive the decomposition for \(\operatorname{SU}(4)\). In particular, we clarify the relationship between various definitions of a Cartan decomposition in the literature and give a complete proof of a general KAK decomposition theorem.
We then distinguish two distinct notions of KAK equivalence classes—double coset equivalence and projective equivalence—thereby addressing mathematical inconsistencies regarding KAK classification in the literature.
Specifically, for \(\operatorname{SU}(4)\), we show that local equivalence classes under multiplication by $\operatorname{SU}(2)\otimes \operatorname{SU}(2)$ are geometrically represented not by the usual ``Weyl chamber'' as claimed in the existing literature. Instead, the ``Weyl chamber'' is only recovered by the projective-local equivalence which disregards global phases.
We develop a systematic theory for determining equivalence and uniqueness for both notions of equivalence.
Our work establishes a rigorous Lie-theoretic foundation for the theory of quantum gates and circuits.
\end{abstract}

\tableofcontents

\section{Introduction}
The KAK decomposition is an important tool in Lie theory, offering a powerful method to decompose a semisimple or reductive Lie group into a product of a compact subgroup $K$, an Abelian subgroup $A$, and the same compact subgroup $K$~\cite{Helgason2001Differential, Knapp1996Lie}. In the context of quantum computing, this decomposition finds its most prominent application in the analysis of two-qubit gates. Specifically, the decomposition of \(\operatorname{SU}(4)\) into the form \(\operatorname{SU}(4) = KAK\), where \(K \cong \operatorname{SU}(2) \otimes \operatorname{SU}(2)\) represents local operations (two single-qubit gates) and \(A\) captures the non-local interactions, provides an indispensable mathematical framework for quantum circuit synthesis, compilation, and classification \cite{Bullock2004Canonical, peterson2020fixed}. Its broad utility has led to its incorporation into most modern quantum programming software packages \cite{pennylane, qiskit, cirq}.

Despite its widespread use, existing accounts of the KAK decomposition in the literature  (see e.g.,~\cite{Zhang:2003zz, dagli2008general, Di2008Cartan, Drury2008Constructive, Earp2005Constructive, Khaneja:2000stb, Perrier2024Solving, wierichsRecursiveCartanDecompositions2025, dalessandroQuantumSymmetriesCartan2007}) have lacked mathematical precision or rigor in several critical respects, mostly relying on heuristic arguments rather than a rigorous Lie-theoretic foundation. More specifically, there exist the following main gaps:
\begin{enumerate}
\item The definition of Cartan decomposition for a Lie algebra was incomplete. In particular, the definiteness conditions were not specified.
\item  The conditions under which the KAK decomposition exists for a given Lie group were not entirely accurate. For instance, the compact and non-compact cases were not distinguished.
\item The definitions and conclusions of the equivalence classes of two-qubit gates under the action of local one-qubit operations were not presented in a mathematically precise manner. In particular, the role of global phases was not considered.
\item The equivalence classes under multiplication by $K$ were not characterized, and their uniqueness was unproven. Specifically for $\operatorname{SU}(4)$, it was not shown that the ``Weyl chamber'' representation is unique.
\item In the mathematical literature, a Weyl chamber is an unbounded cone determined by the reflection hyperplanes of the Weyl group. In quantum computing, however, the compact region used to parametrize these equivalence classes is often also called the ``Weyl chamber'', an inconsistency in the terminology that can be misleading from a Lie-theoretic perspective.
\item The Weyl group acts on the Cartan subalgebra, which is a complex vector space.  It was not shown that the action of the Weyl group can be restricted to the real maximal Abelian vector space.
\end{enumerate}

In this work, we resolve these issues and establish a fully rigorous mathematical foundation for the KAK decomposition and equivalence.  Our main  contributions are summarized as follows.

First, building on the foundational work of Cartan \cite{Cartan1926Classe} and subsequent developments in Lie theory \cite{Knapp1996Lie, Helgason2001Differential}, we give a precise definition of the Cartan decomposition, demonstrate the equivalence of various definitions found in the mathematical literature, and explicitly contrast them with that of the quantum computing literature. We then present a complete and rigorous proof of the KAK decomposition for connected compact semisimple Lie groups (Theorem~\ref{theorem_:kak_theorem_compact_version}). This is a classic theorem when the Lie group is simply-connected, but a complete proof tailored to the needs of quantum information, that is,~a KAK decomposition theorem for connected compact semisimple Lie groups (see \cite{Khaneja:2000stb,dagli2008general,wierichsRecursiveCartanDecompositions2025}), is not found in the literature.

Second, we develop a general theory for determining equivalence classes in a simply-connected compact Lie group \(U\) under the left and right action of a subgroup \(K\). We introduce two distinct notions of equivalence that are relevant in different contexts.
\begin{itemize}
\item \textbf{Double coset equivalence (mathematical version):} This is a standard notion in group theory, where two elements are considered equivalent if they lie in the same double coset \(K u K\). We will show that this classification of equivalence classes is governed by the \(K\)-lattice and the affine Weyl group, and it corresponds to the concept of ``local equivalence''~\cite{Zhang:2003zz, wattsMetricStructureSpace2013}. The definition leads to the quotient space $\overline{Q_0}$, which we call the $T$-cell.

\item \textbf{Projective equivalence (quantum version):} In quantum information theory, two-qubit gates that differ by a global phase are physically indistinguishable. Thus we consider two gates equivalent if they are related by local operations together with an arbitrary global phase. This notion, which we call ``{projective-locally equivalent}'', is governed by the $p$-lattice and the projective affine Weyl group. This definition yields the familiar geometric representation of two-qubit gates known as the ``Weyl chamber,'' which we instead call the $P$-cell to be consistent with the terminology in the mathematical literature.
\end{itemize}

Third, we apply our general theory to \(\operatorname{SU}(4)\). Using the KAK decomposition theorem, we first obtain the explicit decomposition for \(\operatorname{SU}(4)\). We compute its Weyl group and the \(\Sigma\)-lattice (or \(K\)-lattice)  and the corresponding affine Weyl group $\Gamma_{\Sigma}$ (or $\Gamma_{K}$) explicitly, yielding a criterion for determining distinct double coset equivalence classes and presenting a different geometric representation from \cite{Zhang:2003zz}.
Additionally, we compute the $p$-lattice, which leads to the projective affine Weyl group $\Gamma_{p}$. By analyzing the action of the group $\Gamma_p$ on a real maximal Abelian subspace, we obtain a rigorous criterion for classifying projective equivalence classes. This provides a rigorous derivation of the familiar three-dimensional geometric representation where each point in the ``Weyl chamber'' corresponds to a unique projective equivalence class of two-qubit gates.

This paper is organized as follows.
In  Section~\ref{section_:proof_of_kak_theorem}, we present a rigorous proof of the general KAK decomposition theorem for connected compact semisimple Lie groups, clarifying the definitions of a Cartan decomposition.
In Section~\ref{section_:determing_equivalence_classes}, we develop a general theory for determining equivalence classes, distinguishing between the mathematical double coset version and the projective version that disregards global phases.
Section~\ref{section_:application_SU4} applies this theory to \(\operatorname{SU}(4)\), deriving the affine Weyl group, the projective affine Weyl group, and the resulting geometric representations of two-qubit gate equivalence classes with respect to the two different notions of equivalence.

\section{Proof of KAK decomposition theorem}
\label{section_:proof_of_kak_theorem}

Unless otherwise stated, $G$ is a real semisimple Lie group and $\mathfrak{g}_0$ a real semisimple Lie algebra.
First, we define the Cartan decomposition. There are two standard definitions in the literature.
\begin{definition}[\cite{Helgason2001Differential}]\label{definition_:Cartan decomposition-Helgason}
Let \(\mathfrak{g}_0\) be a semisimple Lie algebra over \(\mathbb{R}\), $\mathfrak{g}$ its complexification, \(\sigma\) the conjugation of \(\mathfrak{g}\) with respect to \(\mathfrak{g}_0\).  A vector space direct sum decomposition \(\mathfrak{g}_0=\mathfrak{k}_0 \oplus \mathfrak{p}_0\) of \(\mathfrak{g}_0\) into a Lie subalgebra \(\mathfrak{k}_0\) and a vector subspace \(\mathfrak{p}_0\) is called a \textit{Cartan decomposition} if there exists a compact real form \(\mathfrak{g}_k\) of $\mathfrak{g}$ such that
\[
\sigma \cdot \mathfrak{g}_k \subset \mathfrak{g}_k,
\quad
\mathfrak{k}_0=\mathfrak{g}_0 \cap \mathfrak{g}_k,
\quad
\mathfrak{p}_0=\mathfrak{g}_0 \cap\left(i \mathfrak{g}_k\right).
\]
\end{definition}

\begin{definition}[\cite{Knapp1996Lie}]\label{definition_:Cartan decomposition-Knapp96}
For any real semisimple Lie algebra $\mathfrak{g}_{0}$, a vector space direct sum decomposition\footnote{Note this is a direct sum decomposition for vector spaces, not in the sense of Lie algebras.}
\[
\mathfrak{g}_0=\mathfrak{k}_0 \oplus \mathfrak{p}_0
\]
is called a \textit{Cartan decomposition} if it satisfies
\begin{equation}\label{equation_:commutation_relations_in_Cartan_decomposition}
\left[\mathfrak{k}_0, \mathfrak{k}_0\right] \subset \mathfrak{k}_0, \quad\left[\mathfrak{k}_0, \mathfrak{p}_0\right] \subset \mathfrak{p}_0, \quad\left[\mathfrak{p}_0, \mathfrak{p}_0\right] \subset \mathfrak{k}_0,
\end{equation}
and the Killing form of $\mathfrak{g}_0$ satisfies
\[
B_{\mathfrak{g}_0} \text { is }
\left\{
\begin{array}{l}
\text {negative definite on } \mathfrak{k}_0 \\
\text {positive definite on } \mathfrak{p}_0 \\
\end{array}
\right. .
\]
\end{definition}

We can directly verify that the two definitions are equivalent (cf. Appendix~\ref{appendix_:Proof of the equivalence of two definitions-Cartan decomposition}). We especially note that the commutation relations \eqref{equation_:commutation_relations_in_Cartan_decomposition} are not sufficient: a Cartan decomposition must also satisfy negative and positive definiteness conditions. In previous literature, e.g.,~ \cite{Zhang:2003zz}, the definiteness condition was ignored. In fact, a true Cartan decomposition for the compact Lie algebra $\mathfrak{su}(4)$ satisfying the commutation relations \eqref{equation_:commutation_relations_in_Cartan_decomposition} can only be a trivial one: $\mathfrak{k}_{0} = \mathfrak{su}(4), \mathfrak{p}_{0} = 0$. However, for compact semisimple Lie groups, we prove that even without this definiteness condition, the desired KAK decomposition exists. 

A definition closely related to Cartan decomposition is the Cartan involution. 
This leads to the definition of an \textit{orthogonal symmetric Lie algebra}.
\begin{definition}[(Effective) orthogonal symmetric Lie algebra~\cite{Helgason2001Differential}]\label{definition_:orthogonal symmetric Lie algebra}
For any real Lie algebra $\mathfrak{g}$,
a pair \((\mathfrak{g}, \theta)\) is called an \textit{orthogonal symmetric Lie algebra} if it satisfies
\begin{enumerate}
\item
\(\theta\) is an \textit{involutive  automorphism} of $\mathfrak{g}$, i.e., $\theta^{2} = I, \theta\neq I$.
\item
\(\mathfrak{k}\), the set of fixed points of \(\theta\), is a compactly imbedded subalgebra of $\mathfrak{g}$.
\end{enumerate}
If in addition \(\mathfrak{k} \cap Z_{\mathfrak{g}}=\{0\}\), where \(Z_{\mathfrak{g}}\) denotes the center of $\mathfrak{g}$, then \((\mathfrak{g}, \theta)\) is called \textit{effective}. A pair \((G, K)\) is said to be \textit{associated with} \(\left(\mathfrak{g}_0, \theta\right)\) if \(G\) is a connected Lie group with Lie algebra \(\mathfrak{g}_0\) and \(K\) is a Lie subgroup of \(G\) with Lie algebra \(\mathfrak{k}_0\).
\end{definition}

\begin{example}[\protect{\cite[Example (a-b) in V.1]{Helgason2001Differential}}]\label{example_:orthogonal symmetric Lie algebra}
The following are two examples of orthogonal symmetric Lie algebras.
\begin{enumerate}
\item[a]
 Let \(\mathfrak{g}_0\) be a compact semisimple Lie algebra and \(\theta\) any involutive automorphism of \(\mathfrak{g}_0\). Then \((\mathfrak{g}_0, \theta)\) is an effective orthogonal symmetric Lie algebra (of the compact type).
\item[b]
Let \(\mathfrak{g}_0\) be a noncompact semisimple Lie algebra and let \(\mathfrak{g}_0=\mathfrak{k}_0\oplus\mathfrak{p}_0\) be any Cartan decomposition of \(\mathfrak{g}_0\). Let \(\theta\) denote an involutive automorphism of \(\mathfrak{g}_0\) given by \(\theta(T+X)=T-X (T \in \mathfrak{k}_0, X \in \mathfrak{p}_0)\). Then \((\mathfrak{g}_0, \theta)\) is an effective orthogonal symmetric Lie algebra (of the noncompact type).
\end{enumerate}
\end{example}

Now we give a duality relation relating the compact type and the noncompact type.

\begin{definition}\label{definition_:dual}
Let \((\mathfrak{g}_0, \theta)\) be an orthogonal symmetric Lie algebra. Suppose we have a  vector space direct sum decomposition \(\mathfrak{g}_0=\mathfrak{k}_{0} \oplus \mathfrak{p}_{0}\) that satisfies
\[
\left[\mathfrak{k}_0, \mathfrak{k}_0\right] \subset \mathfrak{k}_0, \quad\left[\mathfrak{k}_0, \mathfrak{p}_0\right] \subset \mathfrak{p}_0, \quad\left[\mathfrak{p}_0, \mathfrak{p}_0\right] \subset \mathfrak{k}_0.
\]
Let $\mathfrak{g}_0^*$ denote the subset \(\mathfrak{k}_0+i \mathfrak{p}_0\) of the complexification \(\mathfrak{g}_0^{\mathbb{C}}\) of \(\mathfrak{g}_0\). With the bracket operation inherited from \(\mathfrak{g}_0^{\mathbb{C}}\), $\mathfrak{g}_0^{*}$ is a Lie algebra over \(\mathbb{R}\). The mapping
\[\theta^{*}: T+i X \rightarrow T-i X \quad (T \in \mathfrak{k}_0, X \in \mathfrak{p}_0) \]
is an involutive automorphism of \(\mathfrak{g}_0^{*}\).
Then \((\mathfrak{g}_0^{*}, \theta^{*})\) is an orthogonal symmetric Lie algebra, called the \textit{dual} of \((\mathfrak{g}_0, \theta)\). In turn, \((\mathfrak{g}_0, \theta)\) is the \textit{dual} of \((\mathfrak{g}_0^{*}, \theta^{*})\).
\end{definition}

With the above definitions, we give the definition of Riemannian symmetric pair.

\begin{definition}[Riemannian symmetric pair~\cite{Helgason2001Differential}]\label{definition_:Riemannian_symmetric_pair}
The pair  \((G, K)\) is said to be a \textit{Riemannian symmetric pair}
if
\begin{enumerate}
\item
\(K\) is closed,
\item
\(\operatorname{Ad}_G(K)\) is compact,
\item
There exists an analytic involutive automorphism \(\tilde{\theta}\) of \(G\) such that
\[\left(G_{\tilde{\theta}}\right)_0 \subset K \subset G_{\tilde{\theta}} \]
where $G_{\tilde{\theta}}$ is the fixed points of $\tilde{\theta}$ in $G$ and $(G_{\tilde{\theta}})_{0}$ is the identity component of $G_{\tilde{\theta}}$.
\end{enumerate}
If only the first and third conditions are satisfied, then $(G,K)$ is said to be a \textit{symmetric pair}.
\end{definition}

As an example, we show that

\begin{proposition}[]\label{proposition_:example_SU_SO_Riemannian symmetric pair}
The pair ($G = \operatorname{SU}(n)$, $K = \operatorname{SO}(n)$) is a Riemannian symmetric pair for all $n \geq 2$.
\end{proposition}

\begin{proof}[Proof]
Check condition 1:  $\operatorname{SO}(n)$ is a compact Lie group defined by polynomial equations ($A^T A = I$, $\det(A)=1$) and is therefore a closed subset of $\operatorname{SU}(n)$.

Check condition 2: The adjoint representation $\operatorname{Ad}: G \rightarrow GL(\mathfrak{g})$ is a continuous map. Since $K = \operatorname{SO}(n)$ is compact, its image $\operatorname{Ad}_G(K)$ is the continuous image of a compact set and is therefore compact.

Check condition 3:
Let $\tilde{\theta}$ be the automorphism of $\operatorname{SU}(n)$ defined by complex conjugation
$$
\tilde{\theta}(g) = \bar{g}.
$$
For any $g \in \operatorname{SU}(n)$, $\bar{g}$ is also in $\operatorname{SU}(n)$. This map is analytic and involutive since  $\tilde{\theta}^2(g) = \bar{\bar{g}} = g$.

For all $g \in \operatorname{SU}(n)$ such that $\tilde{\theta}(g) = g$, i.e., $g = \bar{g}$, this means $g$ is a real matrix. The intersection of $\operatorname{SU}(n)$ (complex unitary matrices) with the real matrices is exactly the orthogonal group $\operatorname{SO}(n)$. So the fixed point set is
$$
G_{\tilde{\theta}} = \{ g \in \operatorname{SU}(n) \mid \bar{g} = g \} = \operatorname{SO}(n).
$$

Since $G_{\tilde{\theta}} = \operatorname{SO}(n)$ is connected for all $n\ge 2$, the identity component is the group itself. Therefore, $(G_{\tilde{\theta}})_0 = \operatorname{SO}(n)$.
Hence we have
\[(G_{\tilde{\theta}})_0 = \operatorname{SO}(n) \subset \operatorname{SO}(n) = K, \]
and
\[K = \operatorname{SO}(n) \subset \operatorname{SO}(n) = G_{\tilde{\theta}}.\]

In sum, the pair ($G = \operatorname{SU}(n)$, $K = \operatorname{SO}(n)$) is indeed a Riemannian symmetric pair for all $n \geq 2$. The proposition follows.
\qedhere
\end{proof}

In the following, let $\mathfrak{g}$ denote a real semisimple Lie algebra for simplicity. We first give a proof of the following theorem.
\begin{theorem}[{KAK decomposition theorem}]
\label{theorem_:KAK_Decomposition_note_version-refined}
Let \(G\) be a connected noncompact semisimple Lie group, \(\mathfrak{g}\)  its Lie algebra.
Let
\[\mathfrak{g}=\mathfrak{k} \oplus \mathfrak{p} \]
be a Cartan decomposition
and let \(\mathfrak{a}\) be an maximal Abelian subspace of \(\mathfrak{p}\).
Then
\[
G=K_{1} A K_{1},
\]
where \(A:=\exp \mathfrak{a}\) and \(K_{1}\) is the connected Lie subgroup of \(G\) corresponding to the Lie algebra \(\mathfrak{k} \subset \mathfrak{g}\).

Moreover, let
\[
\mathfrak{u}:=\mathfrak{k} \oplus i \mathfrak{p}
\]
be the dual real Lie algebra of $\mathfrak{g}$.
Let \(U\) be the connected and simply-connected Lie group with Lie algebra $\mathfrak{u}$. Then
\[
U=K_{2} A^{\prime} K_{2},
\]
where \(A^{\prime}:=\exp i \mathfrak{a}\) and $K_{2}$ is the connected Lie subgroup of \(U\) corresponding to the Lie algebra \(\mathfrak{k} \subset \mathfrak{u}\).
\end{theorem}

\begin{proof}[Proof of Theorem \ref{theorem_:KAK_Decomposition_note_version-refined}]
Since \(G\) is a semisimple Lie group, \(\mathfrak{g}\) is a semisimple Lie algebra.

From Example~\ref{example_:orthogonal symmetric Lie algebra} , given a semisimple Lie algebra \(\mathfrak{g}\) with a Cartan decomposition
\[
\mathfrak{g}=\mathfrak{k} \oplus \mathfrak{p},
\]
and \(\theta\) be the corresponding Cartan involution,  then $(\mathfrak{g}, \theta)$ is an effective orthogonal symmetric Lie algebra. Let
\[\mathfrak{u} = \mathfrak{k} \oplus i \mathfrak{p} \]
and $\theta^{*}$ be the corresponding Cartan involution,
then \((\mathfrak{u}, \theta^{*})\) is the dual orthogonal symmetric Lie algebra.

Since \(G\) is connected and \(\mathfrak{k}\) is stabilized by \(\theta,(G, K)\) is a pair associated with \((\mathfrak{g}, \theta)\).
Furthermore, \(G\) is noncompact by assumption, so \((G, K)\) is a Riemannian symmetric pair. By  \cite[Prop. 3.6,  Ch.IV.3]{Helgason2001Differential}, since \((U, K)\) is a pair associated with the orthogonal symmetric Lie algebra \((\mathfrak{u}, \theta)\), \(U\) is simply-connected and \(K\) is connected, so \((U, K)\) is a Riemannian symmetric pair.

Therefore, by \cite[Theorem 6.7, Ch.V.6]{Helgason2001Differential},
we have
\[
G=K_1 A K_1, \quad U=K_2 A^{\prime} K_2,
\]
where \(A, A^{\prime}\) are the connected Lie subgroups of \(G, U\) associated with \(\mathfrak{a} \subset \mathfrak{g}, i \mathfrak{a} \subset \mathfrak{u}\). Since $\mathfrak a, i\mathfrak a$ are Abelian, by uniqueness $A = \exp \mathfrak{a}$ and $A^{'} = \exp i \mathfrak{a} $.
The conclusion follows.
\qedhere
\end{proof}

Note that from Theorem~\ref{theorem_:KAK_Decomposition_note_version-refined}, we have a direct corollary for simply-connected compact semisimple Lie groups.

\begin{corollary}[]\label{corollary_:kak_decompostion_simply_connected_compact}
Let \(U\) be a connected and simply-connected compact semisimple Lie group, $\mathfrak{u}$ its Lie algebra. If there exists a vector space decomposition
\[
\mathfrak{u} =\mathfrak{k} \oplus \mathfrak{p}
\]
satisfying the commutation relation, i.e.,
\[
[\mathfrak{k}, \mathfrak{k}] \subset \mathfrak{k}, \quad[\mathfrak{p}, \mathfrak{k}] \subset \mathfrak{p}, \quad[\mathfrak{p}, \mathfrak{p}] \subset \mathfrak{k}.
\]
Let $\mathfrak{a}$ be a maximal Abelian subspace of $\mathfrak{p}$. Then
\[
U=K A K,
\]
where \(A:=\exp  \mathfrak{a}\) and $K$ is the connected Lie subgroup of \(U\) corresponding to the Lie algebra \(\mathfrak{k}\).
\end{corollary}

\begin{proof}[Proof]
Given the compact semisimple Lie algebra $\mathfrak{u}$ with the decomposition $\mathfrak{u} = \mathfrak{k} \oplus \mathfrak{p}$ satisfying the commutation relations,
we can define its dual noncompact real form $\mathfrak{g}$ as
$$
\mathfrak{g} := \mathfrak{k} \oplus i\mathfrak{p}.
$$
Since $U$ is compact and semisimple, the Killing form of $\mathfrak u$ is negative definite. Therefore, the Killing form of $\mathfrak{g}$ is negative definite on $\mathfrak{k}$ and positive definite on  $i\mathfrak{p}$. Hence this is a Cartan decomposition of $\mathfrak{g}$. Now let  \(G\) be a connected noncompact semisimple Lie group associated to the Lie algebra \(\mathfrak{g}\).
Then we can apply Theorem~\ref{theorem_:KAK_Decomposition_note_version-refined} to both $G$ and $U$ to get the desired KAK decomposition for $U$.
\qedhere
\end{proof}

Moreover,  we can prove a theorem for compact semisimple Lie groups.
\begin{theorem}[]\label{theorem_:kak_theorem_compact_version}
Let \(U\) be a connected compact semisimple Lie group, $\mathfrak{u}$ its Lie algebra. Suppose there exists a vector space decomposition
\[
\mathfrak{u}:=\mathfrak{k} \oplus \mathfrak{p}
\]
satisfying the commutation relations
\begin{equation}\label{equation_:commutation relations_compact}
[\mathfrak{k}, \mathfrak{k}] \subset \mathfrak{k}, \quad[\mathfrak{p}, \mathfrak{k}] \subset \mathfrak{p}, \quad[\mathfrak{p}, \mathfrak{p}] \subset \mathfrak{k}.
\end{equation}
Let $\mathfrak{a}$ be a maximal Abelian subspace of $\mathfrak{p}$. Then
\[
U=K A K,
\]
where \(A:=\exp  \mathfrak{a}\) and $K$ is the connected Lie subgroup of \(U\) corresponding to the Lie algebra \(\mathfrak{k}\).
\end{theorem}

\begin{proof}[Proof]
Let $\widetilde{U}$  be the universal cover of the connected compact semisimple Lie group $U$ and $\pi: \widetilde{U} \rightarrow U$ be the covering map. By the definition, the Lie algebra $\mathfrak{u}$ of $U$ is the same as that of $\widetilde{U}$, so the decomposition $\mathfrak{u} = \mathfrak{k} \oplus \mathfrak{p}$ lifts to the Lie algebra of $\widetilde{U}$ in the same form.

Since $U$ is compact and semisimple, then its universal cover $\widetilde{U}$ is also compact and semisimple. Applying Corollary~\ref{corollary_:kak_decompostion_simply_connected_compact} to $\widetilde{U} $, we get
$$
\widetilde{U} = \widetilde{K} \widetilde{A} \widetilde{K},
$$
where $\widetilde{A} = \exp(\mathfrak{a})$, and  $\widetilde{K}$ is the connected subgroup of $\widetilde U$ corresponding to $\mathfrak{k} \subset \mathfrak u$.

Since the derivative $\mathrm{d}\pi$ of the covering map $\pi$  preserves the Lie algebra structure, i.e., $\mathrm{d}\pi = \operatorname{id}: \mathfrak{u} \rightarrow\mathfrak{u}$,
and by Proposition~\ref{proposition_:natural property of the exponential map} the exponential map is natural with the Lie group and Lie algebra homomorphisms,
so $\pi$ maps $\widetilde{K}$ to $K$ (the connected subgroup of $U$ with Lie algebra $\mathfrak{k}$), and maps $\widetilde{A}$ to $A = \exp(\mathfrak{a})$ in $U$.
Note that the universal covering map $\pi: \widetilde{U} \to U$ is a Lie group homomorphism, thus the decomposition descends to $U$, which gives
$$
U = \pi(\widetilde{U}) = \pi(\widetilde{K} \widetilde{A} \widetilde{K}) = K A K.
$$
\qedhere
\end{proof}

\section{General theory of determining equivalence classes}
\label{section_:determing_equivalence_classes}

This section provides a general theory for determining equivalence classes. In the mathematical literature,  equivalence is typically understood via double cosets, i.e., given a simply-connected compact Lie group $U$ and a decomposition $U = KAK$, one aims to classify the double coset space $K \backslash U / K$. In quantum information theory, however, two gates are considered equivalent if they differ only by a global phase, which corresponds to a center of the Lie group. This leads to the projective equivalence, double cosets modulo the center, i.e., $(K \backslash U / K) \,/\, Z(U)$, where $Z(U)$ is the center of $U$. Equivalently, one works with the projective group $U / Z(U)$. Therefore, while the double coset version captures the classical Lie-theoretic classification, the projective version incorporates additional symmetry reduction by global phases.
All necessary definitions and notations are summarized in the table below.
\begin{center}
\begin{tabular}{ccc}
\toprule
$(U,K)$ & Riemannian symmetric pair with $U$ compact
\\
$(\mathfrak{u}, \theta)$ & the orthogonal symmetric Lie algebra
\\
$\mathfrak{k}$ & $\theta$'s $+1$-eigenspace in $\mathfrak{u}$
\\
$\mathfrak{p}$ & $\theta$'s $-1$-eigenspace in $\mathfrak{u}$
\\
$\mathfrak{h}$ & an arbitrary maximal Abelian subspace of $\mathfrak{p}$
\\
$\mathfrak{u}^{\mathbb{C}}$ & the complexification of $\mathfrak{u}$
\\
$\mathfrak{h}^{\mathbb{C}}$ &  the complexification of $\mathfrak{h}$
\\
$\alpha$ & the root of $\mathfrak{u}^{\mathbb{C}}$ with respect to $\mathfrak{h}^{\mathbb{C}}$
\\
$\Delta$ & the set of nonzero roots
\\
$\Sigma$ & the set of roots not vanishing identically on $\mathfrak{h}^{\mathbb{C}}$
\\
$W(\Sigma)$ & the group generated by all the reflections
\\
$N_{K}(\mathfrak{h})$ & normalizer of $\mathfrak{h}$ in $K$
\\
$C_{K}(\mathfrak{h}) $ & centralizer of $\mathfrak{h}$ in $K$
\\
$W(U,K)$ & the Weyl group
\\
$D(U, K)$ & the diagram of the pair $(U,K)$
\\
$\mathfrak{h}_r$ & the complement of the diagram
\\
$\mathfrak{a}_e$ & unit lattice
\\
$\mathfrak{a}_K$ & $K$-lattice
\\
$\mathfrak{a}_\Sigma$ & $\Sigma$-lattice
\\
$\mathfrak{a}_p$ & $p$-lattice
\\
$\Gamma_{\Sigma}$ & affine Weyl group
\\
$\Gamma_{p}$ & projective affine Weyl group
\\
\bottomrule
\end{tabular}
\end{center}

We then develop a general method to determine both the double coset equivalence classes and the projective counterpart for simply-connected compact Lie groups.

Let $U$ be a simply-connected compact Lie group, $\mathfrak{u}$ its Lie algebra. Let $(U,K)$ be a Riemannian symmetric pair associated with an orthogonal symmetric Lie algebra $(\mathfrak{u}, \theta)$, where $\mathfrak{u}$ has a vector space direct sum  decomposition $\mathfrak{u} = \mathfrak{k} \oplus \mathfrak{p}$~{We require the Riemannian symmetric pair \( (U,K) \) to be of compact type with \( K \neq U \).  Then the dual \(\mathfrak{k} + i \mathfrak{p} \) automatically satisfies \( \mathfrak{p} \neq \{0\} \), giving a nontrivial Cartan decomposition of noncompact type.}.
Let $\mathfrak{h}$ denote an arbitrary maximal Abelian subspace of \(\mathfrak{p}\).
Let $\mathfrak{u}^{\mathbb{C}}, \mathfrak{h}^{\mathbb{C}}$ be complexification of $\mathfrak{u}, \mathfrak{h}$ respectively.
For \(X \in \mathfrak{h}^{\mathbb{C}}\), let \(W \in \mathfrak{u}^{\mathbb{C}}\) be an eigenvector of \(\operatorname{ad}_X\) and \(\alpha(X)\) the corresponding eigenvalue, i.e.,
\[
[X, W]=\alpha(X) W.
\]
The linear functional \(\alpha\) is called a \textit{root} of \(\mathfrak{u}^{\mathbb{C}}\) with respect to \(\mathfrak{h}^{\mathbb{C}}\).
Let \(\Delta\) denote the set of nonzero roots, and \(\Sigma\) denote the set of roots in \(\Delta\) which do not vanish identically on \(\mathfrak{h}^{\mathbb{C}}\). For each  \(\alpha \in \Sigma\), let \(s_\alpha\) denote the \textit{reflection} with respect to the hyperplane \(\alpha(X)=0\) in \(\mathfrak h\). The group  $W(\Sigma)$ is  generated by all the reflections \(s_\alpha\)  given by roots in $\Sigma$. Each \(\alpha \in \Sigma\)  defines a hyperplane \(\alpha(X)=0\) in the vector space \(\mathfrak{h}\). These hyperplanes divide the space $\mathfrak{h}$ into finitely many connected components, called the \textit{Weyl chambers}. Note that these Weyl chambers are infinitely large. Therefore, to avoid misunderstanding, we use the standard terminology here which is different from that  in~\cite{Zhang:2003zz}.

Let $C_{K}(\mathfrak{h}) $ and  $ N_{K}(\mathfrak{h}) $ denote the \textit{centralizer} and \textit{normalizer} of $\mathfrak{h}$ in \(K\), respectively. In other words,
\[
\begin{aligned}
C_{K}(\mathfrak{h}) & :=  \left\{k \in K \mid \operatorname{Ad}_k(X)=X \text { for each } X \in \mathfrak{h}\right\},  \\
N_{K}(\mathfrak{h}) & :=  \left\{k \in K \mid \operatorname{Ad}_k(\mathfrak{h}) \subset \mathfrak{h}\right\}.
\end{aligned}
\]
We define the \textit{Weyl group} to be the quotient
\begin{equation}\label{equation_:Weyl_group_quotient_version}
W(U,K) : =  N_{K}(\mathfrak{h}) / C_{K}(\mathfrak{h}).
\end{equation}
Since all the maximal Abelian subspaces in $\mathfrak{p}$ are conjugate up to elements in $K$ as a direct corollary of conjugacy of maximal tori in Theorem~\ref{theorem_:Maximal Torus Theorem}, the Weyl group is independent of the choice of arbitrary $\mathfrak{h}$.

To further analyze the group action, first we have a fact about the relationship between the Weyl group \(W(U,K)\) and the group $W(\Sigma)$.

\begin{theorem}[\protect{\cite[Corollary 2.13, Ch.~VII]{Helgason2001Differential}}]\label{theorem_:Weyl_group_identity_of_2_versions}
The Weyl group $W(U,K)$ is generated by the reflections \(s_{\alpha}\), \(\alpha \in \Sigma\). In other words, $W(U,K) = W(\Sigma)$.
\end{theorem}

Next we consider the relationship between the representatives of the double coset\footnote{Two elements are in the same equivalence class if they differ by multiplications of $K$ on the left and right sides. } $K\backslash U/K$ and the orbits of some group acting on the vector space $\mathfrak{h}$. We recall some standard definitions following the notation of~\cite{Helgason2001Differential}.

\begin{definition}[Diagram]\label{definition_:diagram}
The set
\begin{equation}\label{equation_:diagram}
D(U, K): = \left\{H \in \mathfrak{h}: \alpha(H) \in \pi i \mathbb{Z} \quad \text { for some } \alpha \in \Sigma \right\}
\end{equation}
is called the \textit{diagram} of the pair \((U, K)\).
And let \(\mathfrak{h}_r=\mathfrak{h}-D(U, K)\) be the complement of the diagram in $\mathfrak{h}$. The components of \(\mathfrak{h}_{r}\) are called \textit{cells}.
\end{definition}

\begin{definition}[]\label{definition_:affine Weyl group}
The \textit{affine Weyl group}, denoted by \(\Gamma_{\Sigma}\),  is defined to be the group of linear transformations of \(\mathfrak{h}\) generated by the reflections in all the hyperplanes in the diagram \(D(U,K)\).
\end{definition}

\begin{definition}[Unit lattice]\label{definition_:unit_lattice}
The \textit{unit lattice} in \( \mathfrak{h} \) is defined as the set
\[
\mathfrak{a}_{e}=\left\{H \in \mathfrak{h}: \exp H=e\right\},
\]
where $e$ is the identity element in $U$.
\end{definition}

\begin{definition}[$K$-lattice]\label{definition_:K-lattice}
The $K$-lattice is defined as the set
\[
\mathfrak{a}_{K}=\left\{H \in \mathfrak{h}: \exp H \in K\right\}.
\]
\end{definition}

To introduce another lattice, we need an inner product on the vector space $\mathfrak{h}$. For each linear form \(\mu\) on \(\mathfrak{h}\), since the Killing form $B$ on $\mathfrak{u}^{\mathbb{C}}$is non-degenerate, let \(A_{\mu} \in \mathfrak{h}\) be determined by
\[\mu(H)= B\left(A_{\mu}, H\right)
\text{, for all } H \in \mathfrak{h} \]
and define \(\langle\mu, \lambda\rangle: =B\left(A_{\mu}, A_{\lambda}\right)\) for any two such linear forms \(\mu\) and \(\lambda\) on $\mathfrak{h} $.

\begin{definition}[$\Sigma$-lattice]\label{definition_:Sigma_lattice}
The $\Sigma$-lattice is defined as the set spanned by the vectors
\[
\frac{2 \pi i}{\langle\mu, \mu\rangle} A_{\mu} \quad (\mu \in \Sigma),
\]
and let \(\mathfrak{a}_{\Sigma}\) denote the lattice.
\end{definition}

To characterize the structure of the affine Weyl group and the resulting double coset equivalence classes, we use the following results which collectively establish the fundamental relations among the lattices, the affine Weyl group action, and KAK decomposition.

We recall the following important relation between the three lattices (see \cite[Theorem 8.5 of Ch.~VII]{Helgason2001Differential}):
for  \(U / K\) a Riemannian symmetric space of the compact type with \(U\) simply-connected and \(K\) connected, then
\begin{equation}\label{equation_:relations_lattices}
\mathfrak{a}_{K}=\mathfrak{a}_{\Sigma} =\frac{1}{2} \mathfrak{a}_{e} .
\end{equation}

In the following we view the lattice \(\mathfrak{a}_{\Sigma}\) as a group of translations of \(\mathfrak{h}\). Let \(Q_{0}\) be any component of \(\mathfrak{h}_r\) whose closure contains the origin.
By \cite[Theorem 8.3 of Ch.~VII]{Helgason2001Differential}, the affine Weyl group $\Gamma_{\Sigma}$ then admits an inner semidirect product decomposition $\Gamma_{\Sigma} = \mathfrak{a}_{\Sigma} \rtimes W(U,K)$, where $\mathfrak{a}_{\Sigma}$ is normal. Moreover, each orbit of $\Gamma_{\Sigma}$ in $\mathfrak{h}$ intersects the closure $\overline{Q_{0}}$ in exactly one point, i.e., $ \overline{Q_{0}} = \Gamma_{\Sigma}\backslash \mathfrak{h}$.

By \eqref{equation_:relations_lattices}, we replace $\mathfrak{a}_{\Sigma}$ with $\mathfrak{a}_{K}$ in the expression of $\Gamma_{\Sigma}$, which gives
\begin{equation}\label{equation_:affine Weyl group_K}
\Gamma_{\Sigma} = \Gamma_K : = \mathfrak{a}_{K}   \rtimes   W(U,K).
\end{equation}
Now we state our main consequence, which realizes the KAK decomposition for compact-type Riemannian symmetric spaces, identifying the closure of a fixed component \(\overline{Q_0}\) as a representative region for the double coset space \(K \backslash U / K\).
\begin{theorem}[\protecting{\cite[Theorem 8.6 of Ch.~VII]{Helgason2001Differential}}]\label{theorem_:equivalence_class_double_coset}
Let \(U / K\) be a Riemannian symmetric space of the compact type, \(U\) simply-connected, and \(K\) connected. Let \(Q_{0}\) be any component of \(\mathfrak{h}_r\) whose closure contains the origin. Then we have
\[
U=K \exp \overline{Q_0} K,
\]
and each \(u \in U\) can be written as  \(u=k_{1} \exp H k_{2} \left(k_{1}, k_{2} \in K\right)\) with \(H \in \overline{Q_{0}}\) unique.
\end{theorem}

\begin{remark}[]\label{remark_:symmetric_space_definition}
The full definition of a Riemannian symmetric space requires additional geometric structure. For brevity we omit it here; it suffices to know that a Riemannian symmetric pair $(G,K)$ always gives rise to a Riemannian symmetric space $G/K$. In particular, when $U$ is a compact Lie group as in our setting, $(U,K)$ being a Riemannian symmetric pair implies $U/K$ is a Riemannian symmetric space of the compact type. Interested readers may consult  \cite[Ch.~IV]{Helgason2001Differential} for further details.
\end{remark}

The following general definition is motivated by the double coset structure from Theorem~\ref{theorem_:equivalence_class_double_coset} and recovers, for instance, the notion of local equivalence for \(\operatorname{SU}(4)\) appeared in \cite{Zhang:2003zz} as a special case.

\begin{definition}[Locally equivalent]\label{definition_:locally equivalent_general}
Let \(U / K\) be a Riemannian symmetric space of the compact type, \(U\) simply-connected, and \(K\) connected.
Two elements \(u, u_1 \in U\) are called \textit{locally equivalent} if they differ by elements in $K$, i.e.,
\[u = k_1 u_1 k_2, \]
for some \(k_1, k_2 \in K \).
\end{definition}

However, in quantum theory, if two gates only differ by a global phase, they are the same physical operation. This motivates the following definition of \textit{projective-locally equivalent}.

\begin{definition}[Projective-locally equivalent]\label{definition_:projectively_locally_equivalent}
Let \(U / K\) be a Riemannian symmetric space of the compact type, \(U\) simply-connected, and \(K\) connected.
Two elements \(u, u_1 \in U\) are called \textit{projective-locally equivalent} if they differ by elements in $K$ and $Z(U)$, i.e.,
\[u = \ell \cdot k_1 u_1 k_2, \]
for some \(k_1, k_2 \in K \) and $\ell \in Z(U)$.
\end{definition}

\noindent This leads to the following definition.
\begin{definition}[$p$-lattice]\label{definition_:p-lattice}
The \textit{$p$-lattice} is defined to be the set
\[
\mathfrak{a}_{p}=\left\{H \in \mathfrak{h}: \exp H \in Z(U) \cdot K\right\}
\]
where $Z(U)$ is the center of $U$.
\end{definition}

Similarly, viewing $\mathfrak{a}_p$ as a translation of $\mathfrak{h}$, we can define the \textit{projective affine Weyl group} as
\[
\Gamma_{p} = \mathfrak{a}_p   \rtimes W(U,K).
\]
Since $\Gamma_{\Sigma}$ is a subgroup of $\Gamma_{p}$, so the orbit space of the projective affine Weyl group acting on $\mathfrak{h}$ is contained in some $\overline{Q_0}$. As an application of Theorem~\ref{theorem_:equivalence_class_double_coset}, we can prove the following theorem with respect to the orbits of projective affine Weyl group acting on $\mathfrak{h}$. Let \(R_{0}\) be a region which is identical to the quotient space $\Gamma_{p} \backslash \mathfrak{h}$.

\begin{theorem}[]\label{theorem_:equivalence_class_double_coset_projective}
Let \(U / K\) be a symmetric space of the compact type, \(U\) simply-connected, and \(K\) connected.
Let \(R_{0}\) be a region identical to the quotient space $\Gamma_{p} \backslash \mathfrak{h}$.
Then we have
\[
U= Z(U) \cdot K \exp R_0 K,
\]
where $Z(U)$ is the center of $U$. And each \(u \in U\) can be written as
$u= \ell \cdot k_{1} \exp H k_{2}$  $\left(\ell \in Z(U),  k_{1}, k_{2} \in K \right)$
with \(H \in R_{0}\) unique.
\end{theorem}

\begin{proof}[Proof]
Let $u \in U$ and let  \(Q_{0}\) be a fixed component of \(\mathfrak{h}_r\) whose closure contains the origin.  By Theorem~\ref{theorem_:equivalence_class_double_coset},  we can write
$$
u = k_1 \exp H  k_2
$$
with $k_1, k_2 \in K$ and $H \in \Gamma_{\Sigma} \backslash \mathfrak{h}$. By definition of the quotient space and the semi-direct product structure of the affine Weyl group, for some $H_0\in \overline{Q_0}$, $H$ can be expressed as $ r_1 s_1 \cdot H_0$ for all $r_1\in \mathfrak{a}_K, s_1 \in W(U,K)$.

For any $H_1\in \mathfrak{a}_p$, by definition of $p$-lattice, there exists $k_3 \in K, z \in Z(U)$, such that $\exp (H_1) = k_3 \cdot z$.
Hence
\[
u
= z^{-1} \cdot k_1 k_3^{-1} k_3 z\cdot  \exp H  k_2
=z^{-1} \cdot k_1 k_3^{-1} \exp (H_{1})  \exp H  k_2
= z^{-1} \cdot k_1 k_3^{-1} \exp(H_{1}+ H) k_2,
\]
so $H_1 + H = H_1 + r_1 s_1 \cdot H_0 = r_2 r_1 s_1 \cdot H_0$, for some $r_2 \in \mathfrak{a}_p$.
By the arbitrariness of $H_1$, $r_2$ can be any element in $\mathfrak{a}_p$, so $r_2 r_1$ can be any element in $\mathfrak{a}_p$. Moreover, by the semi-direct product structure of the projective affine Weyl group $\Gamma_p = \mathfrak{a}_p  \rtimes W(U,K)$, any element $g\in \Gamma_p$ can be represented as a product $g = r \cdot s $ for some $r \in \mathfrak{a}_p$ and $s\in W(U,K)$ and vice versa. Since $s_1$ can be any element in $W(U,K)$, so $r_2 r_1 s_1 \in \Gamma_p$ and can be any element in $\Gamma_p$.
Hence there exists $H_2 \in \Gamma_{p} \backslash \mathfrak{h} = R_0$, such that
\[
u = z^{-1} \cdot k_1 k_3^{-1} \exp(H_{2}) k_2  \in Z(U) \cdot K \exp(H_2) K.
\]
This proves the existence of such a decomposition of $u\in U$.

For the uniqueness, suppose now that $u$ has two forms of expression,
$$
u = z_1 \cdot r_1 \exp H  r_2
\quad \text{ and }
u = z_2 \cdot l_1 \exp H'  l_2
$$
with $z_1, z_2 \in Z(U), r_1, r_2, l_1, l_2 \in K$ and $H, H' \in R_0$.
By the definition of $p$-lattice, there exists $H_{z_1}, H_{z_2} \in \mathfrak{a}_{p}$, such that $\exp (H_{z_1}) = r_3 \cdot z_1,  \exp (H_{z_2}) = l_3 \cdot z_2$.
Hence
\[
u = z_1 \cdot r_1 \exp H  r_2
= r_1 r_3^{-1} r_3 \cdot z_1 \cdot \exp H  r_2
=  r_1 r_3^{-1} \exp (H_{z_1}) \exp H  r_2
= r_1 r_3^{-1} \exp (H_{z_1}+ H)  r_2,
\]
and similarly
\[
u = l_1 l_3^{-1} \exp (H_{z_2}+ H')  l_2.
\]
Since $r_1 r_3^{-1}, r_2, l_1 l_3^{-1}, l_2 \in K $, by Theorem~\ref{theorem_:equivalence_class_double_coset}, so $H_{z_1}+ H$ and $H_{z_2}+ H'$ are conjugate under the action of the affine Weyl group $\Gamma_{\Sigma}$. Since $H_{z_1}, H_{z_2}\in \mathfrak{a}_p$ and $\Gamma_{p} = \mathfrak{a}_p  \rtimes W(U,K) > \Gamma_{\Sigma}$, so $H$ and $H'$ are conjugate under the action of the projective affine Weyl group $\Gamma_{p}$. Since by assumption $H , H' \in R_0$, this leads to $H = H'$, which proves the uniqueness.
\qedhere
\end{proof}

Hence, given more symmetries on the equivalence relation (Definition~\ref{definition_:projectively_locally_equivalent}), we can restrict our geometry region of different representatives to a smaller one.
To explicitly compute the $p$-lattice for a compact Lie group, generally, it is necessary to handle each case individually, and it turns out to be a non-trivial task.  In the following  Section~\ref{section_:application_SU4}, we will deal with the Riemannian symmetric pair $(U,K) = (\operatorname{SU}(4), \operatorname{SU}(2)\otimes \operatorname{SU}(2))$. We calculate the minimal positive period for each coefficient $c_i$ of elements in $\mathfrak{h}$, which in this case is $\pi/ 2$ and gives the $p$-lattice as defined above.

\section{Key application: $\operatorname{SU}(4)$}
\label{section_:application_SU4}

The KAK theory for $\operatorname{SU}(4)$ is especially important as it serves as the theoretical foundation for two-qubit quantum gates and has broad applications in quantum computing.
In this section, we first apply our above general theory to $\operatorname{SU}(4)$ to derive its KAK decomposition and to calculate the Weyl group using the root system of $\mathfrak{su}(4)$. As an application of Theorem~\ref{theorem_:kak_theorem_compact_version}, we can prove the following theorem.
\begin{theorem}[KAK decomposition for $\operatorname{SU}(4)$]\label{theorem_:KAK_decomposition_for_SU4}
Any element  \(u \in \operatorname{\operatorname{SU}}(4)\) can be written in the following form
\[
u = k_1 a k_2=k_1 \exp \left\{i\left(c_1 \sigma_x^1 \sigma_x^2 + c_2 \sigma_y^1 \sigma_y^2 + c_3 \sigma_z^1 \sigma_z^2 \right)\right\} k_2
\]
where \(k_1, k_2 \in \operatorname{SU}(2) \otimes \operatorname{SU}(2)\), $c_1,c_2,c_3 \in \mathbb{R}$,   \(\sigma_x, \sigma_y, \sigma_z\) are the Pauli matrices, and
\(\sigma_\alpha^1 \sigma_\beta^2 := \sigma_\alpha^1 \otimes \sigma_\beta^2\).
\end{theorem}

\begin{proof}[Proof]
First note that the Lie group $\operatorname{SU}(4)$ is compact and semisimple. A basis of $\mathfrak{g} = \mathfrak{su}(4)$ in the form of Pauli matrices is given by the following,
\begin{equation}\label{equation_:basis_of_kp_in_su4_in_pauli_matrix_}
\begin{aligned}
\mathfrak{k}
= &
i \left\{ \sigma_x^1, \sigma_y^1, \sigma_z^1, \sigma_x^2, \sigma_y^2, \sigma_z^2 \right\}, \\
\mathfrak{p}
= &
i \left\{ \sigma_x^1 \sigma_x^2, \sigma_x^1 \sigma_y^2, \sigma_x^1 \sigma_z^2, \sigma_y^1 \sigma_x^2, \sigma_y^1 \sigma_y^2, \sigma_y^1 \sigma_z^2, \sigma_z^1 \sigma_x^2, \sigma_z^1 \sigma_y^2, \sigma_z^1 \sigma_z^2 \right\}.
\end{aligned}
\end{equation}
Note that the basis vectors of $\mathfrak{k}$ are shorthand as they omit the tensor product with an identity operator.
We can check by definition that $\mathfrak{k}, \mathfrak{p}$ satisfy the commutation relations (cf. Appendix~\ref{appendix_:Proof of Commutation relations}),
\[ [\mathfrak{k}, \mathfrak{k}] \subset \mathfrak{k}, \quad[\mathfrak{p}, \mathfrak{k}] \subset \mathfrak{p}, \quad[\mathfrak{p}, \mathfrak{p}] \subset \mathfrak{k}.  \]
And the real Lie subalgebra
\[\mathfrak{a} = \langle i \sigma_x^1 \sigma_x^2, i \sigma_y^1 \sigma_y^2, i \sigma_z^1 \sigma_z^2 \rangle_{\mathbb{R}} \]
forms  a maximal Abelian subspace of $\mathfrak{p}$. So applying Theorem~\ref{theorem_:kak_theorem_compact_version}, we get
\[
\operatorname{SU}(4)=K A K,
\]
where \(A = \exp  \mathfrak{a}\) and $K$ is the connected Lie subgroup of $\operatorname{SU}(4)$ corresponding to the Lie algebra \(\mathfrak{k} \subset \mathfrak{su}(4)\), which is indeed $\operatorname{SU}(2) \otimes \operatorname{SU}(2)$. So the theorem holds.
\qedhere
\end{proof}

By Proposition~\ref{proposition_:example_SU_SO_Riemannian symmetric pair}, we know that $(\operatorname{SU}(4), \operatorname{SO}(4))$ is a Riemannian symmetric pair.
Since $\operatorname{SO}(4)$ is isomorphic to $ \operatorname{SU}(2)\otimes\operatorname{SU}(2)$ under the Bell basis transformation (see below in the proof of Proposition~\ref{proposition_:minimum positive period_SU4} or in Appendix~\ref{appendix_:Bell basis transformation}), so
\[(U,K) = (\operatorname{SU}(4), \operatorname{SU}(2)\otimes\operatorname{SU}(2)) \]
is a Riemannian symmetric pair.

First, we compute the Weyl group \(W(\Sigma)\).
Let $X_{i}$ be elements from $\mathfrak{k}$ and $\mathfrak{p}$ given above from the first to the last, i.e.,
\[
\begin{aligned}
&X_1 = i \sigma_x^1 , & & X_2 = i \sigma_y^1 , & & X_3 = i \sigma_z^1 , \\
&X_4 = i \sigma_x^2 , & & X_5 = i \sigma_y^2 , &&X_6 = i \sigma_z^2 , \\
&X_7 = i \sigma_x^1 \sigma_x^2 ,& &X_8 = i \sigma_x^1 \sigma_y^2 , &&X_9 = i \sigma_x^1 \sigma_z^2 , \\
&X_{10} = i \sigma_y^1 \sigma_x^2 ,& &X_{11} = i \sigma_y^1 \sigma_y^2 , & &X_{12} = i \sigma_y^1 \sigma_z^2 , \\
&X_{13} = i \sigma_z^1 \sigma_x^2 , & &X_{14} = i \sigma_z^1 \sigma_y^2 , & &X_{15} = i \sigma_z^1 \sigma_z^2 .\\
\end{aligned}
\]
Let
\begin{equation}\label{equation:X_c1c2c3}
X=i\left(c_1 \sigma_x^1 \sigma_x^2+c_2 \sigma_y^1 \sigma_y^2+c_3 \sigma_z^1 \sigma_z^2\right) \in \mathfrak{a}.
\end{equation}
Since the Killing form of $\mathfrak{su}(4)$ is given by
\[B(X, Y) = 8 \cdot \operatorname{tr}(X Y), \]
for any $X,Y \in \mathfrak{su}(4)$, so it satisfies the following property,
\[
B(X_{i}, X_{j}) = - 32 \delta_{ij}.
\]
By identifying \(\mathfrak{a}\) with the Euclidean space \(\mathbb{R}^3\), we can write $X$ in~\eqref{equation:X_c1c2c3} as a column vector \(\left[c_1, c_2, c_3\right]\).
Let $\mathfrak{g}^{\mathbb{C}}$ and $\mathfrak{a}^{\mathbb{C}}$ be the complexification of $\mathfrak{g}$ and   $\mathfrak{a}$ respectively. Since every Cartan subalgebra of $\mathfrak{su}(4)^{\mathbb{C}}$ has dimension $3$, so $\mathfrak{a}^{\mathbb{C}}$ is a Cartan subalgebra, thus the root system of the $\mathfrak{g}^{\mathbb{C}}$ is given by the following vector space direct sum decomposition
\begin{equation}\label{equation_:root system of su_4}
\mathfrak{g}^{\mathbb{C}} = \mathfrak{a}^{\mathbb{C}} \oplus \bigoplus_{\alpha, \beta \in \{x,y,z\}, \alpha \neq \beta} \mathfrak{g}^{\mathbb{C}}_{\pm e_{\alpha} \pm e_{\beta}},
\end{equation}
where $e_{\alpha}$ are linear functionals satisfying
\[
e_{\alpha} (i \sigma_{\gamma}^{1} \sigma_{\gamma}^{2}) = 2 i \delta_{\alpha \gamma} = \left\{
\begin{array}{ll}
2i & \quad \alpha = \gamma,
\\
0 & \quad \alpha \neq \gamma.
\end{array}
\right.
\]
The computation details of the root system and associated linear functionals are provided in Appendix~\ref{appendix_:Calculation of the root system}.
\begin{proposition}[\cite{Zhang:2003zz}]\label{proposition_:reflections}
For \(X = \left[c_1, c_2, c_3\right]\in \mathfrak{a}\), the reflections of $X$ through \(s_{\pm e_{\alpha} \pm e_{\beta}}\) are given as follows,
\begin{equation}\label{equation_zhang_15:}
\begin{aligned}
s_{e_{z} -e_{y}}(X) & =\left[c_1, c_3, c_2\right], \quad s_{e_{y}+ e_{z}}(X)=\left[c_1,-c_3,-c_2\right],
\\
s_{ e_{y} - e_{x}}(X) & =\left[c_2, c_1, c_3\right], \quad s_{e_{x}+ e_{y}}(X)=\left[-c_2,-c_1, c_3\right],
\\
s_{e_{x} - e_{z}}(X) & =\left[c_3, c_2, c_1\right], \quad s_{e_{x}+ e_{z}}(X)=\left[-c_3, c_2,-c_1\right].
\end{aligned}
\end{equation}
And the Weyl group $W(U,K) = W(\Sigma)$ is generated by these reflections.
\end{proposition}

\begin{proof}[Proof]
We can identify \(\mathfrak{a}\) with \(\mathbb{R}^3\). Since the inner product $\langle\cdot, \cdot \rangle$ in $\mathfrak{a}$ is induced from the Killing form $B$ and so is given by the standard inner product up to a constant scalar in the $3$-dimensional Euclidean space.
For the root $e_x - e_z$, so the plane \((e_x - e_z)(Y)=0\) in \(\mathfrak{a}\) is the set
\[
\Pi_{e_x - e_z} =
\left\{Y \in \mathbb{R}^3 \mid u^\top Y=0\right\},
\]
where \(u=[1,0,-1]\).
The reflection of \(X=\left[c_1, c_2, c_3\right]\) with respect to the plane \(\Pi_{e_x - e_z}\) is
\[
s_{e_x - e_z}(X)= X- \frac{2 \langle u,X\rangle}{\|u\|^2} u =\left[c_3, c_2, c_1\right].
\]
Similarly, we can compute all the other reflections \(s_{\pm e_{\alpha} \pm e_{\beta}}\), which gives \eqref{equation_zhang_15:}.
\qedhere
\end{proof}

Therefore, we have a direct corollary,
\begin{corollary}[]\label{corollary_:Weyl group structure}
Each element in the Weyl group is equivalent to either a permutation of the entries of \(\left[c_1, c_2, c_3\right]\), or a permutation with sign flips of two entries. In other word, \(\left[c_1, c_2, c_3\right]\) can be transformed by elements in the Weyl group into forms like
\[[c_i, c_j,c_k], [c_i, - c_j,- c_k], [-c_i, c_j,-c_k],[-c_i, -c_j,c_k] \]
where $(i,j,k)$ is an arbitrary permutation of $(1,2,3)$.
\end{corollary}

\begin{remark}[]\label{remark_:}
The Weyl group defined by the root system acts on the complex vector space $\mathfrak{a}^{\mathbb{C}}$. But here from Proposition~\ref{proposition_:reflections}, all those reflections preserve the real vector space $\mathfrak{a}$. So the Weyl group action can be restricted to  $\mathfrak{a}$. Hence we can talk about the action of the Weyl group on this $3$-dimensional Euclidean space as in Theorem~\ref{theorem_:equivalence_class_double_coset} and Theorem~\ref{theorem_:equivalence_class_double_coset_projective}. This is a technical detail that is missing in quantum computing literature~\cite{Zhang:2003zz}.
\end{remark}

By the KAK decomposition Theorem~\ref{theorem_:KAK_decomposition_for_SU4}, we have the following expression for any \(u \in \operatorname{SU}(4)\),
\[
u = k_1 a k_2=k_1
\exp \left\{i\left(c_1 \sigma_x^1 \sigma_x^2+c_2 \sigma_y^1 \sigma_y^2+c_3 \sigma_z^1 \sigma_z^2\right)\right\}
k_2.
\]
Note that the expression for $a$ is periodic in each $c_{k}$ since an obvious result is $2\pi$ in each $c_i$. In the next two sections, we will do detailed calculations to find the period and use the results to determine corresponding equivalence classes.

\subsection{Double coset equivalence classes}
\label{subsection_:double_coset_equivalence_classes}
In this subsection, we derive the double coset equivalence classes. This corresponds to the affine Weyl group $\Gamma_{\Sigma}$. Since we have already calculated the Weyl group, so next we need to calculate the $K$-lattice according to Definition~\ref{definition_:K-lattice}. By \eqref{equation_:relations_lattices}, we only need to calculate the unit lattice.

\begin{proposition}[]\label{proposition_:unit lattice of SU4}
The unit lattice $\mathfrak{a}_e$ is given by
\begin{equation}\label{equation_:unit lattice_SU4}
\mathfrak{a}_{e} =
\left\{[c_1, c_2, c_3]\in \mathbb{R}^{3}
\middle|
\begin{aligned}
&c_i \in \pi\mathbb{Z}, \\
&c_i / \pi \text{ are all even or} \\
& \text{only two of them are odd. }
\end{aligned}
\right\}.
\end{equation}
\end{proposition}

\begin{proof}[Proof]
Since $\sigma_x^1 \sigma_x^2, \sigma_y^1 \sigma_y^2, \sigma_z^1 \sigma_z^2$ commute with each other, we can simultaneously diagonalize them using the Bell basis
\[
\begin{aligned}
\frac{1}{\sqrt{2}}(|00\rangle+|11\rangle),&
\frac{i}{\sqrt{2}}(|01\rangle+|10\rangle),
\\
\frac{1}{\sqrt{2}}(|01\rangle-|10\rangle),&
\frac{i}{\sqrt{2}}(|00\rangle-|11\rangle).
\end{aligned}
\]
In the Bell basis, the matrix of \(U\in \operatorname{SU}(4)\) can be written as
\[
U_B=Q^{\dagger} U Q,
\]
where
\[
Q=\frac{1}{\sqrt{2}}\left(\begin{array}{cccc}
1 & 0 & 0 & i \\
0 & i & 1 & 0 \\
0 & i & -1 & 0 \\
1 & 0 & 0 & -i
\end{array}\right).
\]
Explicitly, the matrices of  operators $\sigma_x^1 \sigma_x^2, \sigma_y^1 \sigma_y^2, \sigma_z^1 \sigma_z^2$ in the Bell basis are given as follows,

\[
\begin{pmatrix}
1 & 0 & 0 &  0  \\
0 & 1 & 0 & 0  \\
0 & 0 & -1 & 0  \\
0 & 0 & 0 & -1
\end{pmatrix},
\begin{pmatrix}
-1 & 0 & 0 &  0  \\
0 & 1 & 0 & 0  \\
0 & 0 & -1 & 0  \\
0 & 0 & 0 & 1
\end{pmatrix},
\begin{pmatrix}
1 & 0 & 0 &  0  \\
0 & -1 & 0 & 0  \\
0 & 0 & -1 & 0  \\
0 & 0 & 0 & 1
\end{pmatrix}.
\]
Since $A = \exp(\mathfrak{a})$ is a compact Lie subgroup of $U$, by Proposition~\ref{proposition_:exponential map compact onto} the exponential map $\exp : \mathfrak{a}\rightarrow A$ is onto, so the operator
\[
\exp \left\{i\left(c_1 \sigma_x^1 \sigma_x^2 + c_2 \sigma_y^1 \sigma_y^2+c_3 \sigma_z^1 \sigma_z^2 \right)\right\}
\]
can be expressed in the matrix form
\begin{equation}\label{equation_:diagonal_A}
\exp \left(
i
\begin{pmatrix}
c_{1} - c_{2}+ c_{3} & 0 & 0 &  0  \\
0 & c_{1} + c_{2} - c_{3} & 0 & 0  \\
0 & 0 & - c_{1} - c_{2} - c_{3} & 0  \\
0 & 0 & 0 & -c_{1} + c_{2}+ c_{3}
\end{pmatrix}
\right)
\end{equation}
in the Bell basis.

By definition of unit lattice,
\[
\exp \left\{i\left(c_1 \sigma_x^1 \sigma_x^2 + c_2 \sigma_y^1 \sigma_y^2+c_3 \sigma_z^1 \sigma_z^2 \right)\right\}  = I_4
\]
 if and only if  the elements
\[
c_{1} - c_{2}+ c_{3}, \quad
c_{1} + c_{2} - c_{3} , \quad
- c_{1} - c_{2} - c_{3}, \quad
-c_{1} + c_{2}+ c_{3}
\]
are all in the set $2\pi \mathbb{Z}$.
Adding each pair of them gives $c_i \in \pi\mathbb{Z}$.
Now, all four elements divided by $\pi$ are in $2\mathbb{Z}$,
so the only possibilities are that $c_i/\pi$ are all even or only two of them are odd, which gives the unit lattice $\mathfrak{a}_e$.
\qedhere
\end{proof}

As a corollary, we can get the $K$-lattice using \eqref{equation_:relations_lattices},
\begin{proposition}[]\label{proposition_:K-lattice_SU4}
The $K$-lattice $\mathfrak{a}_K$ is given by
\begin{equation}\label{equation_:K-lattice_SU4}
\mathfrak{a}_{K} =
\left\{[c_1, c_2, c_3]\in \mathbb{R}^{3}
\middle|
\begin{aligned}
&2 c_i \in \pi\mathbb{Z}, \\
&2 c_i / \pi \text{ are all even or} \\
& \text{only two of them are odd. }
\end{aligned}
\right\}
\end{equation}
\end{proposition}

The diagram of the $K$-lattice in the $c_1 c_2$-plane is given in Figure~\ref{fig:klattice}. It is the same for other two planes ($c_2 c_3$-plane, $c_3 c_1$-plane). In the $3$-dimensional space $c_1c_2c_3$,  the fundamental period (minimum positive period) is $\pi$ along each coordinate direction ($c_1, c_2, c_3$).

\begin{figure}
\centering
\includegraphics[width=0.8\linewidth]{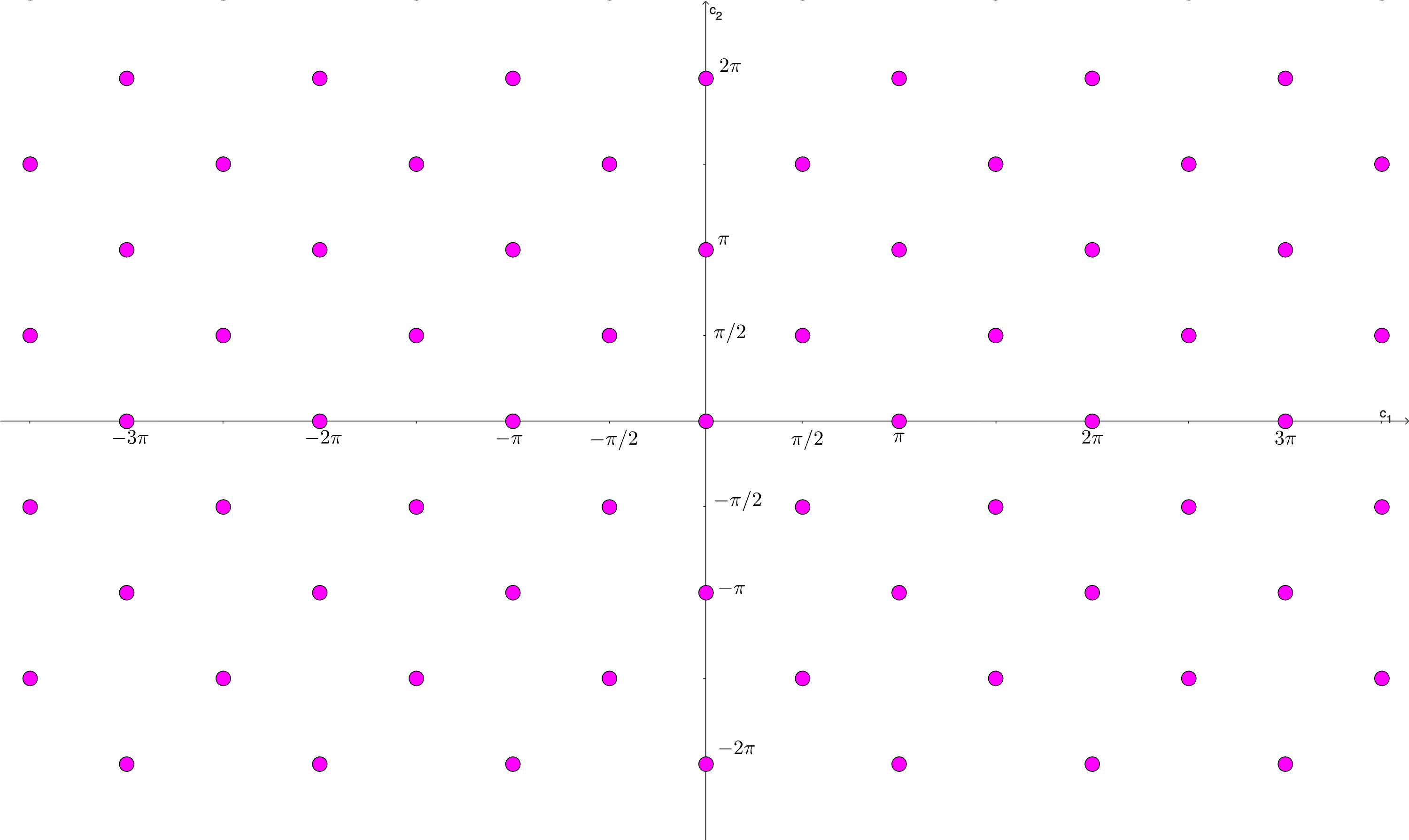}
\caption{$K$-lattice (in the $c_1 c_2$-plane).}
\label{fig:klattice}
\end{figure}

Combining the Weyl group,  where the reflections are either permutations or permutations with sign flips of two entries in \(\left[c_1, c_2, c_3\right]\) from \eqref{equation_zhang_15:},  we get the following criterion for determining distinct  double coset equivalence classes.

\begin{proposition}[]\label{proposition_:equivalence_relation_double coset}
If \(\left[c_1, c_2, c_3\right]\) is a coordinate of \([u]\) as an double coset equivalence class, then
\begin{equation}\label{equation_:double_coset_equivalence_class_relation}
\begin{aligned}
& \left[c_i, c_j, c_k\right], \left[c_i+ \pi, c_j, c_k\right], \left[c_i+ \pi/2, c_j+ \pi/2, c_k\right], \\
& \left[\pi/2-c_i, \pi/2-c_j, c_k\right],\left[c_i, \pi/2-c_j, \pi/2-c_k\right], \left[\pi/2-c_i, c_j, \pi/2-c_k\right]
\end{aligned}
\end{equation}
are also coordinates for \([u]\), where \((i, j, k)\) can be any permutation of \((1,2,3)\). And \eqref{equation_:double_coset_equivalence_class_relation} generates the full set of equivalence relations for double coset equivalence class.
\end{proposition}

Therefore, we can reduce the symmetry in the $3$-dimensional Euclidean space using these equivalences to get a geometric representation of the double coset equivalence classes.

\begin{proposition}[]\label{proposition_:criterion_equivalence_classes_double_coset}
Let $\vec{c} = [c_1, c_2, c_3]$ be a coordinate for  $u\in \operatorname{SU}(4)$,  we can transform \(\vec{c}\) using \eqref{equation_:double_coset_equivalence_class_relation} such that it lies in
\begin{equation}\label{equation_:geometric_representation-tetrahedral cell_2}
T:=
\left\{
\vec{c} \in \mathbb{R}^3
\middle|
 \pi/2>c_1 \geq c_2 \geq |c_3| \geq 0
, c_1+c_2 \leq \pi/2
\right\}.
\end{equation}
And each element in $T$ represents a unique double coset equivalence class.
\end{proposition}

We call the above region $T$ a \textit{tetrahedral cell} (or $T$-cell) and its graph is given in Figure~\ref{fig:tetrahedral-cell2}. The vertices are $A = (\pi/4, \pi/4, \pi/4), B = (\pi/2, 0, 0 ), U = (\pi/4, \pi/4, -\pi/4)$. Tetrahedrons $OABC$ and $OUBC$ are inside the top and bottom cubes respectively and $A,U$ are the center points of each cube. The tetrahedron $OABU$ shows the region $T$ and each point represents a unique equivalence class.

\begin{figure}[H]
\centering
\includegraphics[width=1\linewidth]{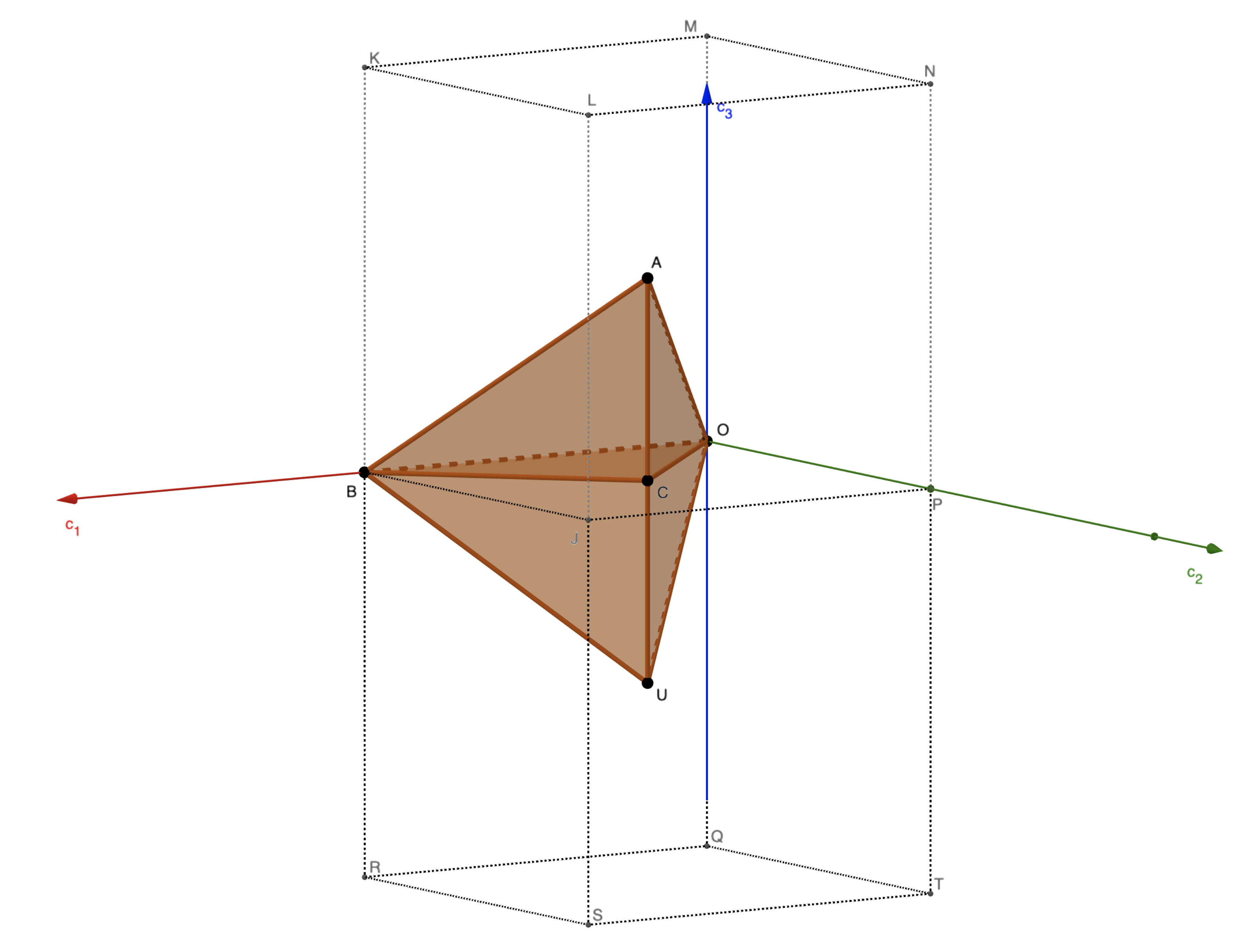}
\caption{Tetrahedral cell ($T$-cell)  given by \eqref{equation_:geometric_representation-tetrahedral cell_2}.}
\label{fig:tetrahedral-cell2}
\end{figure}

\begin{proof}[Proof of Proposition~\ref{proposition_:criterion_equivalence_classes_double_coset}]
For \(\vec{c} = [c_1, c_2, c_3]  \in \mathbb{R}^3\). We implement the following algorithm.

Via translation relations
\[
\left[c_i, c_j, c_k\right] \sim \left[c_i, c_j, c_k+\pi\right] \sim \left[c_i+ \pi/2, c_j+ \pi/2, c_k\right],
\]
we can make  \(c_1, c_2 \in\left[0, \pi/2\right)\) and $c_3\in [-\pi/2,\pi/2)$.
Via permutations of each pair of entries, we can make \(c_1 \geq c_2 \geq c_3\).

If \(c_1+c_2> \pi/2 \), then transform using
\[
[c_1, c_2, c_3] \sim\left[\pi/2- c_2, \pi/2-c_1, c_3\right].
\]
If \(c_3\) is larger than any of \(\pi/2-c_2, \pi/2-c_1\), permute those two  entries again. Note that \(\pi/2-c_2+c_3 \leq \pi/2\). By renaming them we get
\[
c_1+c_2 \leq \pi/2.
\]

If $c_2+c_3 <0$, transform using
\[
[c_1, c_2, c_3] \sim [c_1, -c_3, -c_2],
\]
then  $-c_2-c_3 > 0 $ and $-c_3 > -c_2$. Hence we can make $c_1 \geq c_2 \geq |c_3|$.

In sum, we can transform $\vec{c}$ into the region given in \eqref{equation_:geometric_representation-tetrahedral cell_2}.

As for uniqueness, by \eqref{equation_:diagram} and \eqref{equation_:root system of su_4}, we get
\[
D(U, K) = \left\{H \in \mathfrak{h}: (\pm e_{\alpha} \pm e_{\beta})(H) \in \pi i \mathbb{Z} \quad \text { for some } \alpha, \beta \in \{x,y,z\}, \alpha \neq \beta \right\},
\]
 which is identical to the following set (since $H = c_1 \sigma_x^1 \sigma_x^2 + c_2 \sigma_y^1 \sigma_y^2+c_3 \sigma_z^1 \sigma_z^2$)
\[
 \left\{[c_1, c_2, c_3] \in \mathbb{R}^{3}:
(\pm c_\alpha \pm c_\beta) \in \frac{\pi}{2}  \mathbb{Z} \quad \text { for some } \alpha, \beta \in \{1,2,3\}, \alpha \neq \beta \right\},
\]
In $D(U, K)$, the enclosed region of four planes
\[
\Pi_1 : c_1 - c_2 = 0, \quad \Pi_2 : c_1 + c_2 = \pi/2, \quad \Pi_3 : c_2 - c_3 = 0, \quad \Pi_4 : c_2 + c_3 = 0
\]
is the tetrahedron $OABU$ as shown in Figure~\ref{fig:tetrahedral-cell2}.
Moreover, no other planes pass through the interior of the region. Hence, the tetrahedron $OABU$ is such a $\overline{Q_0}$ in Theorem~\ref{theorem_:equivalence_class_double_coset}, which implies the uniqueness of each point in the tetrahedral cell $T$.
\qedhere
\end{proof}

\subsection{Projective equivalence classes}
\label{subsection_:projective_equivalence_classes}
In this subsection, we derive the projective equivalence classes.
We need to determine the smallest element such that the image of the exponential map in $A$ lies in $K$  up to global phases. In other words, we need to calculate the following \textit{minimum positive period} for $\operatorname{SU}(4)$.

\begin{definition}[Minimum positive period]\label{definition_:minimum positive period}
We define the \textit{minimum positive period}  to be the smallest positive number in each $c_{i}$ such that
\[
\exp \left\{
i\left(c_1 \sigma_x^1 \sigma_x^2 + c_2 \sigma_y^1 \sigma_y^2+c_3 \sigma_z^1 \sigma_z^2 \right)
\right\}
\in
Z(U) \cdot K
\]
where $Z(U)$ is the center of $\operatorname{SU}(4)$.
\end{definition}

We have the following main result.
\begin{proposition}[]\label{proposition_:minimum positive period_SU4}
The minimum positive period for $\operatorname{SU}(4)$ is $\pi/2$. Moreover, the $p$-lattice is given by
\begin{equation}\label{equation_:p_lattice_SU4}
\mathfrak{a}_p  = \{ [c_1, c_2, c_3] \in \mathbb{R}^{3} | c_i \in \pi \mathbb{Z}/2 \}.
\end{equation}
\end{proposition}

\begin{proof}[Proof]
Recall that in the proof of Proposition~\ref{proposition_:unit lattice of SU4} the operator
\[
\exp \left\{i\left(c_1 \sigma_x^1 \sigma_x^2 + c_2 \sigma_y^1 \sigma_y^2+c_3 \sigma_z^1 \sigma_z^2 \right)\right\}
\]
can be expressed as
\begin{equation}\label{equation_:diagonal_A}
\exp \left(
i
\begin{pmatrix}
c_{1} - c_{2}+ c_{3} & 0 & 0 &  0  \\
0 & c_{1} + c_{2} - c_{3} & 0 & 0  \\
0 & 0 & - c_{1} - c_{2} - c_{3} & 0  \\
0 & 0 & 0 & -c_{1} + c_{2}+ c_{3}
\end{pmatrix}
\right)
\end{equation}
in the Bell basis.

Now, transforming $\mathfrak{k}$ defined in~\eqref{equation_:basis_of_kp_in_su4_in_pauli_matrix_} into the Bell basis we get the following operators
\[i Q^{\dagger}\left\{\sigma_x^1, \sigma_y^1, \sigma_z^1, \sigma_x^2, \sigma_y^2, \sigma_z^2\right\} Q, \]
which forms a real basis for \(\mathfrak{so(4)}\) viewed as a vector subspace of real matrices. Applying the exponential map to $\mathfrak{so}(4)$, since $\operatorname{SO}(4)$ is compact, by Proposition~\ref{proposition_:exponential map compact onto}, the exponential map $\exp: \mathfrak{so}(4)\rightarrow\operatorname{SO}(4)$ is onto, so we get the real Lie group $\operatorname{SO}(4)$.
The  complete set of diagonal matrices in $\operatorname{SO}(4)$ is given by
\[\operatorname{diag}(\pm 1, \pm 1, \pm 1, \pm 1) \text{ with an even number of } -1 \text{'s}.\]
By adding global phases $\{\pm i , \pm 1\}$, these diagonal matrices, which are elements of $\operatorname{SU}(4)$, can only be
\[\operatorname{diag}(\pm 1, \pm 1, \pm 1, \pm 1) \text{ with an even number of } -1 \text{'s}\]
or
\[\operatorname{diag}(\pm i, \pm i, \pm i, \pm i) \text{ with an even number of } -i \text{'s}.\]
Therefore, a necessary condition is that the matrices have the form
\begin{equation}\label{equation_:diagonal_in_K_global_phase}
\exp \left\{ \frac{\pi i }{2}
\operatorname{diag}( b_{1},  b_{2},  b_{3},  b_{4})
\right\}, \quad \text{where $b_{1},  b_{2},  b_{3},  b_{4}\in \mathbb{Z}$ have the same parity.  }
\end{equation}
Combining the two expressions \eqref{equation_:diagonal_A} and \eqref{equation_:diagonal_in_K_global_phase}, the elements
\[
c_{1} - c_{2}+ c_{3}, \quad
c_{1} + c_{2} - c_{3} , \quad
- c_{1} - c_{2} - c_{3}, \quad
-c_{1} + c_{2}+ c_{3}
\]
 divided by $\pi/2$ are integers and have the same parity.
Hence, by adding each pair of the elements we conclude that  $c_{1}, c_{2}, c_{3} \in \pi \mathbb{Z}/2$.

On the other hand, note that we have the following equality
\[
\exp\{i\frac{\pi}{2} \sigma_x^1 \sigma_x^2\} = -i\cdot \exp\{i \left(\frac{\pi}{2} \sigma_x^1 + \frac{\pi}{2} \sigma_x^2 \right)\} \in iK,
\]
we can exactly achieve the minimal positive period $\pi/2$ in the entry $c_1$ and it is the same for other two entries $c_2$ and $c_3$.
In sum, the minimum positive period in each $c_{i}$ is $\pi/2$.

Moreover, since $c_{1}, c_{2}, c_{3} \in \pi \mathbb{Z}/2$, so the minimum positive period indeed gives the $p$-lattice
\[
\mathfrak{a}_p  = \{ [c_1, c_2, c_3] \in \mathbb{R}^{3} | c_i \in \pi \mathbb{Z}/2 \}.
\]

\qedhere
\end{proof}

The diagram of the $p$-lattice in the $c_1 c_2$-plane is given as below (Figure~\ref{fig:qlattice}). Note that we now have points such as $(\pi/2, 0,0), (0,\pi/2,0)$, unlike the $K$-lattice in  Figure~\ref{fig:klattice}. The diagram is the same for other two planes ($c_2$-plane, $c_3$-plane). And in the $3$-dimensional space $c_1c_2c_3$, the minimum period is $\pi/2$ along each coordinate direction.
\begin{figure}[H]
\centering
\includegraphics[width=0.8\linewidth]{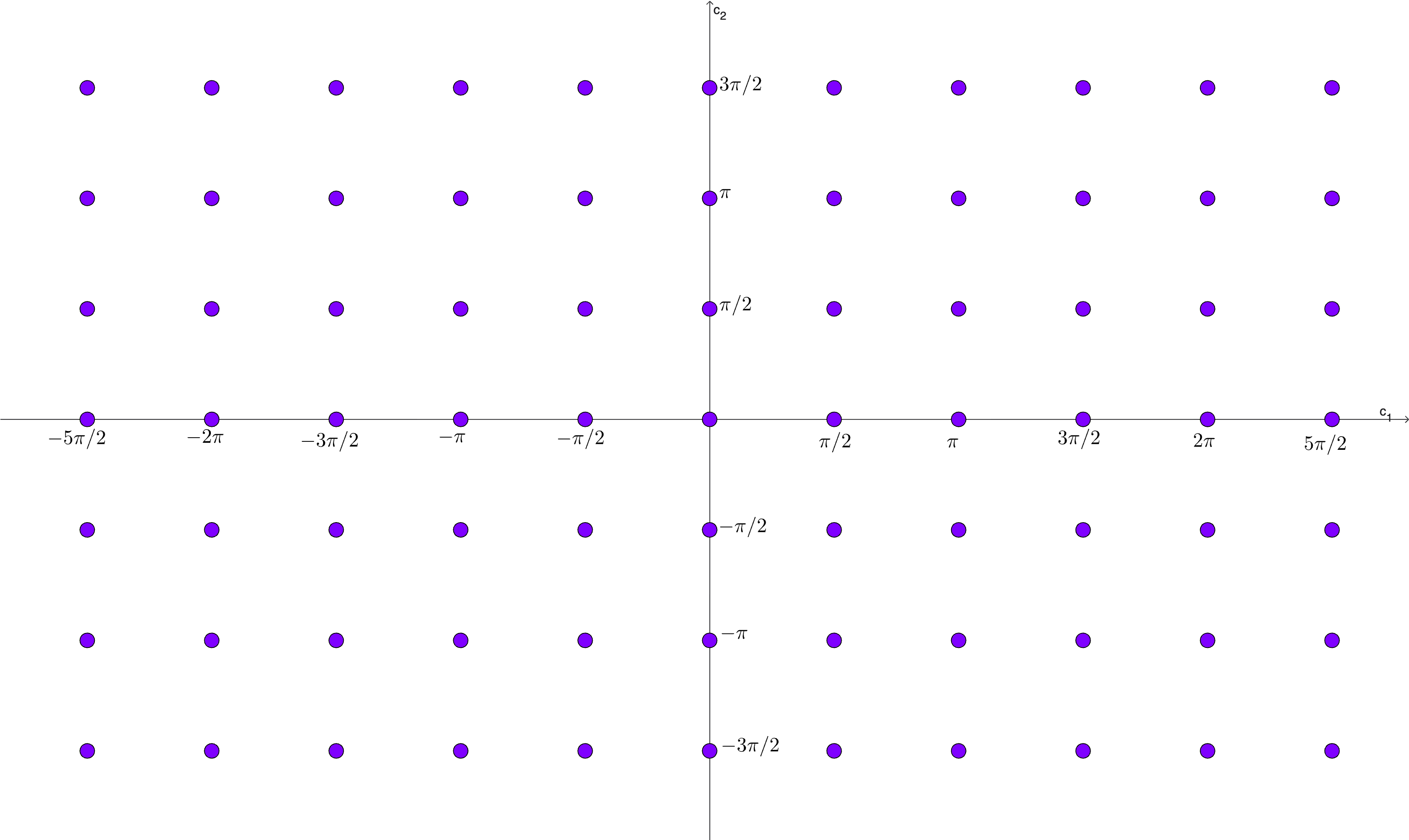}
\caption{$p$-lattice (in the $c_1 c_2$-plane).}
\label{fig:qlattice}
\end{figure}

Combining the Weyl group, using Theorem~\ref{theorem_:equivalence_class_double_coset_projective},  we get the following criterion for determining distinct projective equivalence classes.
\begin{proposition}[]\label{proposition_:equivalence_relation_global_phases}
If \(\left[c_1, c_2, c_3\right]\) is a coordinate of \([u]\) for $u\in \operatorname{SU}(4)$ as an projective equivalence classes, then
\begin{equation}\label{equation_:projective equivalence classes}
\begin{aligned}
& \left[c_i, c_j, c_k\right], \left[c_i+ \pi/2, c_j, c_k\right],
\\
& \left[\pi/2-c_i, \pi/2-c_j, c_k\right],\left[c_i, \pi/2-c_j, \pi/2-c_k\right], \text{and} \left[\pi/2-c_i, c_j, \pi/2-c_k\right]
\end{aligned}
\end{equation}
are also coordinates for \([u]\), where \((i, j, k)\) can be any permutation of \((1,2,3)\). And \eqref{equation_:projective equivalence classes} generates the full set of equivalence relations for projective equivalence class.
\end{proposition}

Therefore, we can reduce the symmetry in the $3$-dimensional Euclidean space using the above equivalences to get the geometric representation of the projective equivalence classes. This is the set of projective-locally equivalent two-qubit gates.

\begin{proposition}[]\label{proposition_:criterion_equivalence_classes_up_to_global_phases}
Let $\vec{c} = [c_1, c_2, c_3]$ be a coordinate of $[u]$ for  $u\in \operatorname{SU}(4)$,  we can transform \(\vec{c}\) using \eqref{equation_:projective equivalence classes} so that it lies in
\begin{equation}\label{equation_:geometric_representation-projective tetrahedral cell}
P:=
\left\{
\vec{c} \in \mathbb{R}^3
\left\lvert\,
\pi/2 > c_1 \geq c_2 \geq c_3 \geq 0
, c_1+c_2 \leq \pi/2, \text { if } c_3=0, \text { then } c_1 \leq \pi/4.
\right.
\right\}.
\end{equation}
And each element in $P$ represents a unique projective equivalence class.
\end{proposition}

\begin{proof}[Proof of Proposition~\ref{proposition_:criterion_equivalence_classes_up_to_global_phases}]
Consider \(\vec{c} = [c_1, c_2, c_3]  \in \mathbb{R}^3\). We implement the following algorithm.

Via translation, we can make \(c_1, c_2, c_3 \in\left[0, \pi/2\right)\).
Via permutations of each pair of entries, we can make \(c_1 \geq c_2 \geq c_3\).

If \(c_1+c_2> \pi/2 \), then transform using
\[
[c_1, c_2, c_3] \sim\left[\pi/2- c_2, \pi/2-c_1, c_3\right].
\]
If \(c_3\) is larger than any of \(\pi/2-c_2, \pi/2-c_1\), permute those two  entries again. Note that \(\pi/2-c_2+c_3 \leq \pi/2\). By renaming them we get
\[
c_1+c_2 \leq \pi/2.
\]

Finally if \(c_3=0\) and \(c_1>\pi/4\), then transform using
\[
[c_1, c_2, 0] \sim\left[ \pi/2-c_1, c_2, 0\right].
\]
We get \eqref{equation_:geometric_representation-projective tetrahedral cell} as desired.

As for uniqueness, by Theorem~\ref{theorem_:equivalence_class_double_coset_projective}, we only need to geometrically determine the quotient space $R_0 = \Gamma_{p}\backslash \mathfrak{h}$. Since the affine Weyl group $\Gamma_{\Sigma}$  is a subgroup of the projective affine Weyl group $\Gamma_{p}$, hence $R_0$ is contained in one such $\overline{Q_{0}}$. We choose $\overline{Q_{0}}$ to be the $T$-cell shown in Figure~\ref{fig:tetrahedral-cell2}. Since center of $\operatorname{SU}(4)$ is $Z(\operatorname{SU}(4)) = \{\pm 1, \pm i\}$ and  $\pm 1 \in K = SU(2)\otimes SU(2), \pm i \notin K$, and by Definition~\ref{definition_:projectively_locally_equivalent}, we only need to consider pairs of two points in $T$-cell such that the corresponding unitaries differ by $i$. For an arbitrary $U\in\operatorname{SU}(4)$, it has a unique representative $[c_1, c_2, c_3]$ in the $T$-cell by Proposition~\ref{proposition_:criterion_equivalence_classes_double_coset}, and similarly for $iU \in\operatorname{SU}(4)$. We observe the following transformation
\[
[c_1, c_2, c_3] \rightarrow [\pi/2 - c_1, c_2, -c_3]
\]
leads to a scalar multiplication by $i$ up to multiplication by elements in $K$ on both sides, so the unique coordinate for $iU$ in the $T$-cell is $[\pi/2 - c_1, c_2, -c_3]$.
By identifying the two points $[c_1, c_2, c_3]$ and $[\pi/2 - c_1, c_2, -c_3]$, we get the region $P$ given by  \eqref{equation_:geometric_representation-projective tetrahedral cell}. This implies that each point inside $P$ represents a unique projective equivalence class. We can also directly prove uniqueness via the argument in Appendix~\ref{appendix_:more_direct_proof_of_uniqueness}.

\qedhere
\end{proof}

Combining Theorem~\ref{theorem_:equivalence_class_double_coset_projective}, Theorem~\ref{theorem_:KAK_decomposition_for_SU4} and  Proposition~\ref{proposition_:criterion_equivalence_classes_up_to_global_phases}, we can get our main result of the KAK decomposition and the projective equivalence class classification for $\operatorname{SU}(4)$.
\begin{theorem}[]\label{theorem_:}
For any element  \(u \in \operatorname{\operatorname{SU}}(4)\), there exist \(k_1, k_2 \in \operatorname{SU}(2) \otimes \operatorname{SU}(2)\), $\ell \in \{\pm1,\pm i\}$, and a unique $[c_1,c_2,c_3] \in P$, where
\[P = \left\{
[c_1,c_2,c_3] \in \mathbb{R}^3
\left\lvert\,
\begin{aligned}
&\pi/2 > c_1 \geq c_2 \geq c_3 \geq 0, c_1+c_2 \leq \pi/2, \\
&\text {and if } c_3=0, \text { then } c_1 \leq \pi/4.
\end{aligned}
\right.
\right\},\]
such that
\[
u = \ell \cdot k_1 \exp \left\{i\left(c_1 \sigma_x^1 \sigma_x^2 + c_2 \sigma_y^1 \sigma_y^2 + c_3 \sigma_z^1 \sigma_z^2 \right)\right\} k_2.
\]
\end{theorem}

We call the above region $P$ a \textit{projective tetrahedral cell} (or \textit{$P$-cell}), which is shown in Figure~\ref{fig:projective-tetrahedral-cell}. The coordinates of vertices are $A = (\pi/4, \pi/4, \pi/4), B = (\pi/2, 0, 0 ), C = (\pi/4, \pi/4,0), S = (\pi/4, 0,0)$. The tetrahedron $OABC$ is inside the cube and $A$ is the center point of the cube. For $\triangle OBC$ (base of the tetrahedron), $\triangle OCS$ is equivalent to $\triangle BCS$ reflected across $\overline{SC}$. And each of the rest points represents a distinct projective equivalence class.

\begin{figure}[H]
\centering
\includegraphics[width=1\linewidth]{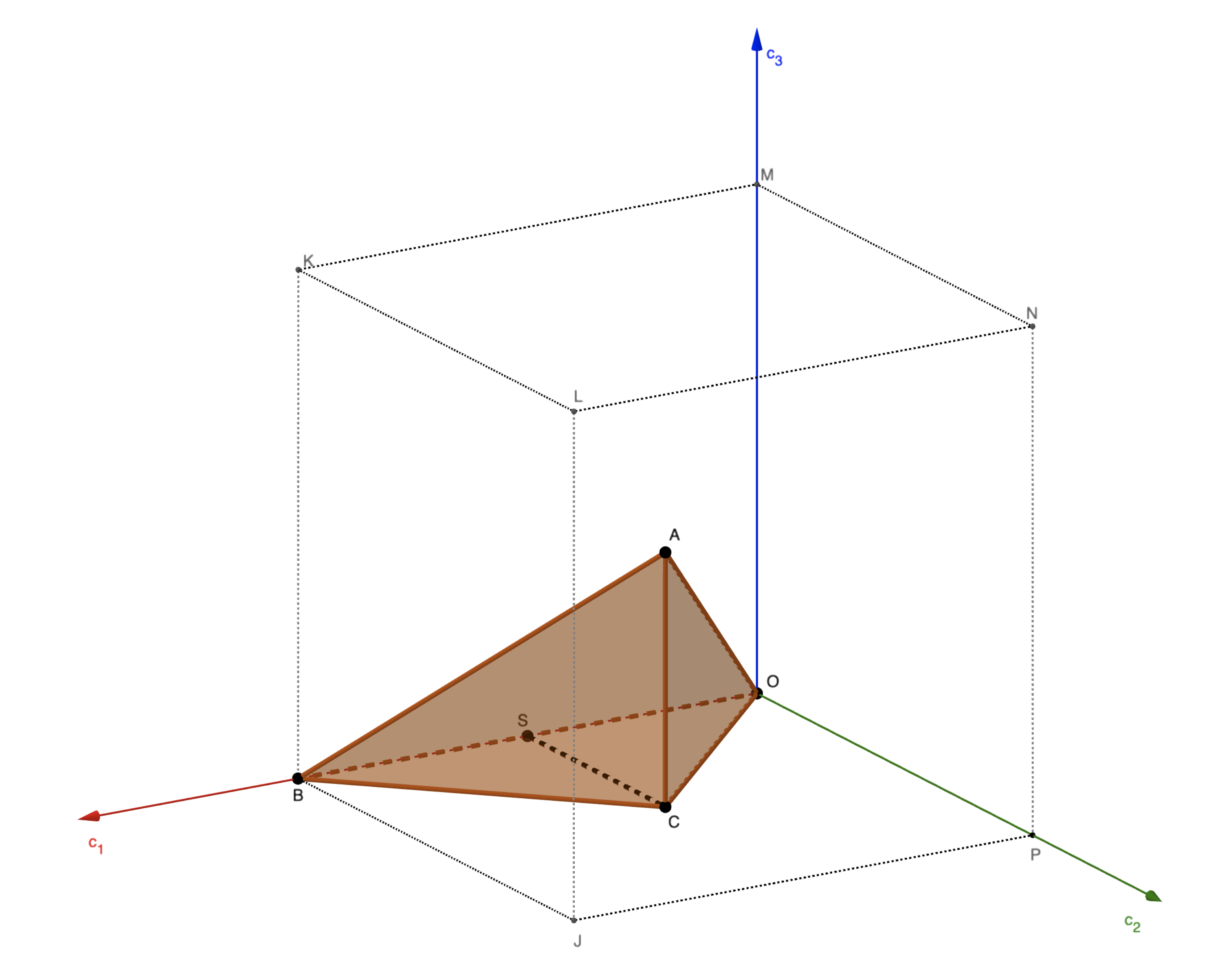}
\caption{Projective tetrahedral cell ($P$-cell).}
\label{fig:projective-tetrahedral-cell}
\end{figure}

\subsection{Comparison between the equivalence classes}

To conclude this section, we offer a comparative remark on the two types of equivalence classes introduced above. Motivated by practical applications in quantum information theory, quantum gates that differ only by a global phase are considered physically equivalent. This leads to the notion of projective equivalence classes, as stated in Definition~\ref{definition_:projectively_locally_equivalent}.

In \ref{subsection_:double_coset_equivalence_classes} and \ref{subsection_:projective_equivalence_classes}, we have explicitly computed two distinct lattice structures corresponding to these two equivalence relations. The difference is evident from Figure~\ref{fig:klattice} and Figure~\ref{fig:qlattice}. Building on these lattices, we have constructed two different associated groups, the affine Weyl group and the projective affine Weyl group. By examining the action of each group on the maximal Abelian subspace $\mathfrak{h}$, we obtained the geometric regions shown in~Figure~\ref{fig:tetrahedral-cell2} and  Figure~\ref{fig:projective-tetrahedral-cell} which show a complete set of distinct equivalence classes.

In the previous literature~\cite{Zhang:2003zz}, the geometric region shown in Figure~\ref{fig:projective-tetrahedral-cell} is typically referred to as the ``Weyl chamber''. However, we adopt the term projective tetrahedral cell, because in standard mathematical terminology, the Weyl chamber is infinitely large and arises from the action of the Weyl group. In our case, the relevant action is that of the projective affine Weyl group. Hence, a correction in terminology is needed.

Despite the apparent differences between these two geometric representations, they are closely related. For instance, a gate represented by a point inside  the tetrahedron $OABC$ in Figure~\ref{fig:tetrahedral-cell2} multiplied by the global phase $i$ admits a counterpart inside the tetrahedron $OBCU$. This correspondence is realized by a reflection along the line $c_1 = \pi/4, c_3 = 0$, i.e.,
\[
[c_1, c_2, c_3] \sim [\pi/2 - c_1, c_2, -c_3].
\]
Algebraically, this reflects the fact that certain central elements like $\pm i I_4$ are absent from the group $K$, and thus this symmetry cannot be eliminated solely by the double coset equivalence relations.

Finally, we provide two figures that illustrate the locations of various representative instances of two-qubit gates (which we also view as representative elements in $\operatorname{SU}(4)$) widely used in quantum information and computation in the $T$-cell and $P$-cell, respectively.
For example, $\operatorname{SWAP}$ gate exchanges the states of the two qubits, whose matrix in the computational basis is
$$
\operatorname{SWAP} =
\begin{pmatrix}
1 & 0 & 0 & 0 \\
0 & 0 & 1 & 0 \\
0 & 1 & 0 & 0 \\
0 & 0 & 0 & 1
\end{pmatrix}.
$$
Its coordinate in the $T$-cell is $[\pi/4, \pi/4,\pi/4]$, thus the same $[\pi/4, \pi/4,\pi/4]$ in the $P$-cell. Similarly, we can compute other two-qubit gates' coordinates in the $T$-cell as well as in the $P$-cell, e.g.,
\begin{center}
\begin{tabular}{cc}
\toprule
Elements in $\operatorname{SU}(4)$  & Coordinates in the $P$-cell (and $T$-cell) \\
\midrule
$\operatorname{I}$ & $[0,0,0]$
\\
$\operatorname{SWAP}$ & $[\pi/4, \pi/4,\pi/4]$
\\
$\sqrt{\operatorname{SWAP}}$ &  $[\pi/8, \pi/8,\pi/8]$
\\
$\operatorname{iSWAP}$ & $[\pi/4, \pi/4,0]$
\\
$\sqrt{\operatorname{iSWAP}}$ &  $[\pi/8, \pi/8,0]$
\\
$\operatorname{CNOT}$ & $[\pi/4, 0,0]$
\\
$\operatorname{B}$ & $[\pi/4, \pi/8,0]$
\\
$\operatorname{QFT}$ & $[\pi/4, \pi/4, \pi/8]$
\\
$\chi$ & $\left[\pi/4-\tfrac{1}{8}\arccos(1/5),\ \pi/8,\ \tfrac{1}{8}\arccos(1/5)\right]$\\
\bottomrule
\end{tabular}
\end{center}
The $\chi$-gate is the unique two-qubit gate that exactly generates a $3$-design sandwiched by single-qubit Haar-random gates recently identified in \cite{kongConvergenceEfficiencyQuantum2024}.
The locations of these representative gates in the $P$-cell are shown in Figure~\ref{fig:projectivetetrahedralcellwithgate}.
\begin{figure}[H]
\centering
\includegraphics[width=0.75\linewidth]{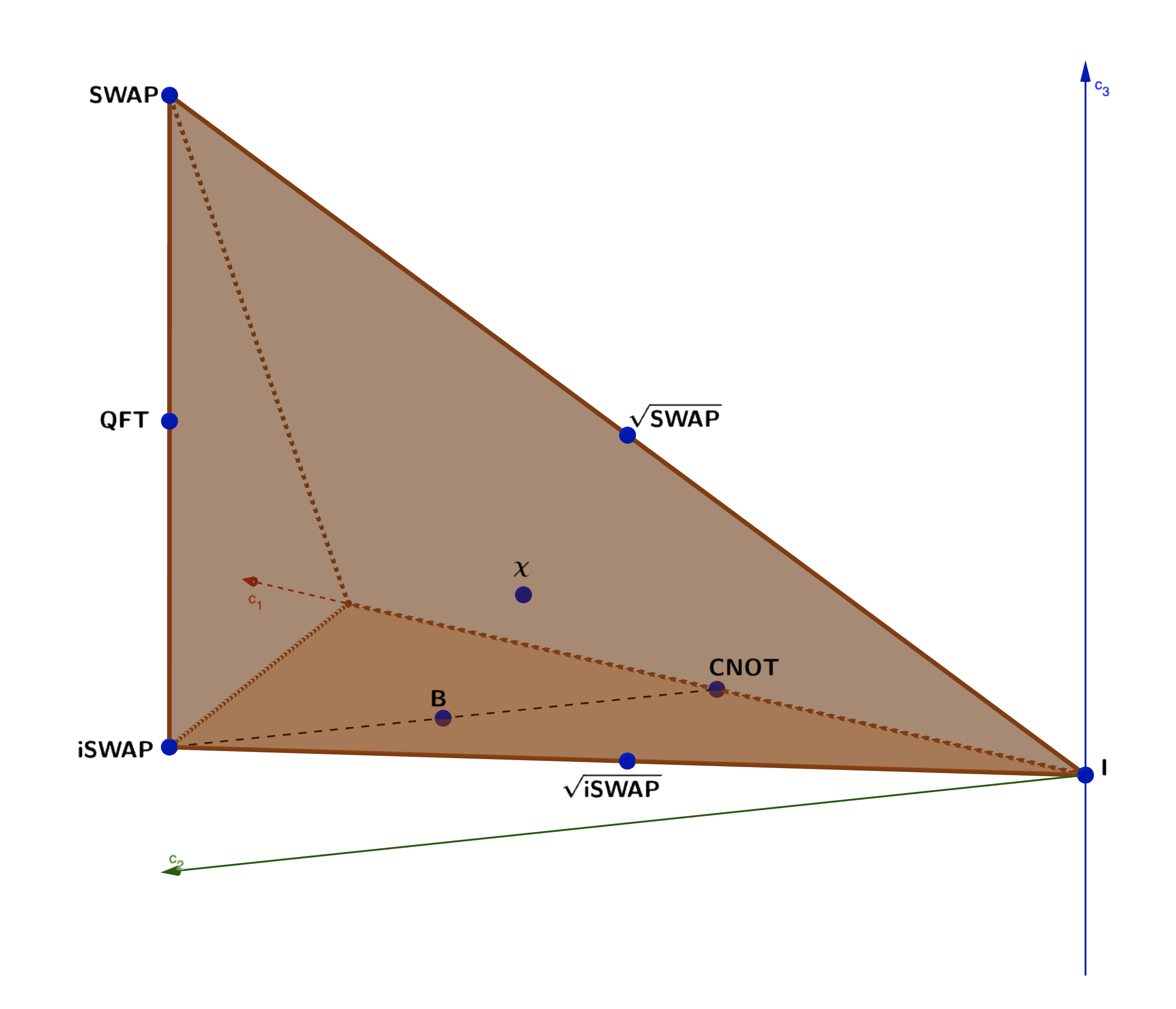}
\caption{Representative two-qubit gates in the $P$-cell.}
\label{fig:projectivetetrahedralcellwithgate}
\end{figure}

In contrast, recall that the $T$-cell is a twice as large region that additionally distinguishes different global phases.
For instance, define $\operatorname{SWAP}'$ to be the above $\operatorname{SWAP}$-gate multiplying a global phase $i$, which represents the action that exchanges the states of the two qubits and attaches a global phase $i$. Its matrix in the computational basis is
$$
\operatorname{SWAP}' = i\cdot
\begin{pmatrix}
1 & 0 & 0 & 0 \\
0 & 0 & 1 & 0 \\
0 & 1 & 0 & 0 \\
0 & 0 & 0 & 1
\end{pmatrix},
$$
and its coordinate in the $T$-cell is $[\pi/4, \pi/4,-\pi/4]$. Similarly, we define $\operatorname{I}'$,  $\sqrt{\operatorname{SWAP}}'$, $\sqrt{\operatorname{\operatorname{iSWAP}}}'$, $\operatorname{QFT}'$ and  $\chi'$ to be $I$, $\sqrt{\operatorname{SWAP}}$, $\sqrt{\operatorname{\operatorname{iSWAP}}}$, $\operatorname{QFT}$ and  $\chi$ multiplying a global phase $i$ respectively. Their coordinates in the $T$-cell are
\begin{center}
\begin{tabular}{cc}
\toprule
Elements in $\operatorname{SU}(4)$  & Coordinates in the $T$-cell \\
\midrule
$\operatorname{I}'$ & $[\pi/2,0,0]$
\\
$\operatorname{SWAP}'$ & $[\pi/4, \pi/4, -\pi/4]$
\\
$\sqrt{\operatorname{SWAP}}'$ &  $[3\pi/8, \pi/8, -\pi/8]$
\\
$\operatorname{iSWAP}'$ & $[\pi/4, \pi/4,0]$
\\
$\sqrt{\operatorname{iSWAP}}'$ &  $[3\pi/8, \pi/8,0]$
\\
$\operatorname{CNOT}'$ & $[\pi/4, 0,0]$
\\
$\operatorname{B}'$ & $[\pi/4, \pi/8,0]$
\\
$\operatorname{QFT}'$ & $[\pi/4, \pi/4, -\pi/8]$
\\
$\chi'$ & $\left[\pi/4+\tfrac{1}{8}\arccos(1/5),\ \pi/8,\ -\tfrac{1}{8}\arccos(1/5)\right]$\\
\bottomrule
\end{tabular}
\end{center}
Figure~\ref{fig:tetrahedralcellwithgate} presents the locations of those two-qubit gates with their counterparts which are by multiplication of a global phase $i$. For example, $\operatorname{SWAP}$ and $\operatorname{SWAP}'$ represent two different elements in the group $\operatorname{SU}(4)$ in which they differ by multiplication of a global phase $i$. Therefore, the points representing $\operatorname{SWAP}$ and $\operatorname{SWAP}'$ are symmetric along the line $c_1 = \pi/4, c_3 = 0$.
Since they are considered the same in the sense of projective equivalence, so they lie at the same point ($\operatorname{SWAP}$) in the $P$-cell as shown in Figure~\ref{fig:projectivetetrahedralcellwithgate}.
The situation is the same for such pairs, e.g.,~  $\sqrt{\operatorname{\operatorname{iSWAP}}}$ and $\sqrt{\operatorname{\operatorname{iSWAP}}}'$, $\operatorname{QFT}$ and $\operatorname{QFT}'$, $\chi$ and $\chi'$, where the relationship between them is also by multiplication of a global phase $i$. Note that the coordinates of $\operatorname{iSWAP}'$, $\operatorname{CNOT}'$ and $\operatorname{B}'$ are the same as those of $\operatorname{iSWAP}$, $\operatorname{CNOT}$ and $\operatorname{B}$, for brevity we only draw one out of each pair.
\begin{figure}[H]
\centering
\includegraphics[width=0.75\linewidth]{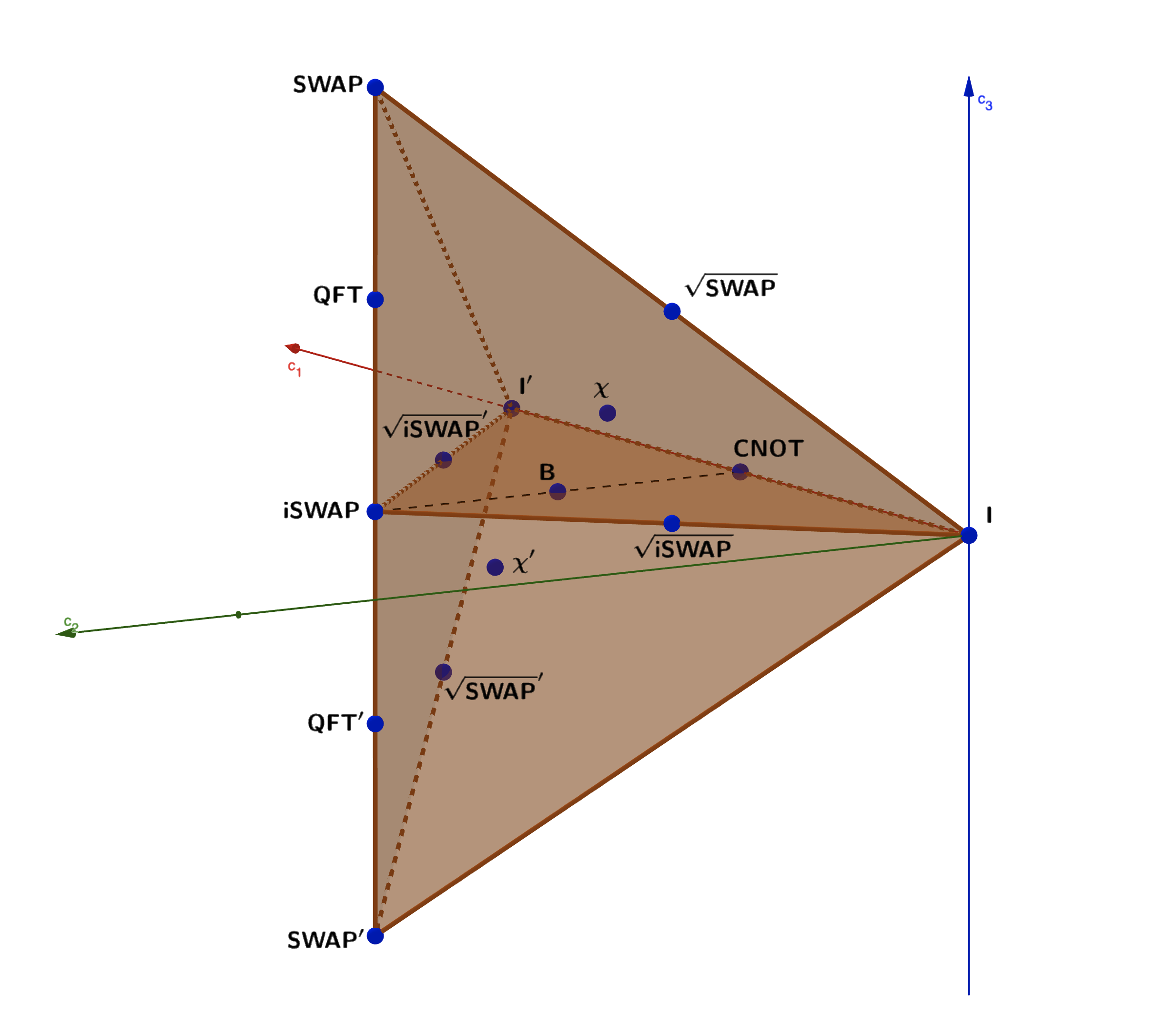}
\caption{Representative two-qubit gates in the $T$-cell. Note that the same gate can split across different locations due to global phase ambiguities.}
\label{fig:tetrahedralcellwithgate}
\end{figure}

\section{Conclusions and outlook}
In this work, we closed notable gaps in the Lie-theoretic foundations of the KAK decomposition and its applications, especially in quantum computing. Specifically, we make the heuristic arguments in the literature fully rigorous. Our primary aim, as set forth in the introduction, was twofold. First, we provide a clear and rigorous mathematical framework for understanding both the KAK decomposition for connected compact Lie groups and the different notions of equivalence classes for simply-connected compact Lie groups.  Second, we apply this general framework to \(\operatorname{SU}(4)\) to derive a precise and rigorous justification for the familiar tetrahedral representation of two-qubit gates.

We summarize our key findings again.
First, by establishing the KAK decomposition theorem for connected compact semisimple Lie groups, we clarified the conditions for its existence.
The resulting theorem (Theorem \ref{theorem_:kak_theorem_compact_version}) provides a solid foundation for its application to \(\operatorname{SU}(4)\) and beyond.
 Second, we developed a general theory for classifying group elements under two distinct notions of equivalence. By distinguishing between the double coset equivalence and the physically motivated projective equivalence, we derive the two corresponding groups, the affine Weyl group and its projective counterpart. This theoretical distinction is crucial for correctly interpreting the geometric spaces used in two-qubit gate classifications.
Third, we applied this general method to \(\operatorname{SU}(4)\), explicitly computing its Weyl group and deriving the \(K\)-lattice, which leads to the affine Weyl group. Additionally, we calculated the $p$-lattice and the corresponding projective affine Weyl group. By analyzing the actions of these two groups, we obtained two distinct geometric representations, the \(T\)-cell and the \(P\)-cell, and presented proofs of uniqueness for distinct points in each cell. This calculation rigorously justifies the familiar geometric representation of two-qubit gates, where the \(P\)-cell (previously referred to as the ``Weyl chamber'' in the literature) instead of the $T$-cell serves as a fundamental domain for projective-locally equivalent two-qubit gates.

Our work points to several interesting open directions for future research.
First, we have primarily stated and proved the KAK decomposition theorem for the compact case. It is well-established that analogous decompositions exist for a wider class of reductive Lie groups. A question for future investigation is whether the decomposition \(G = KAK\) holds for any Lie group whose Lie algebra admits a vector space decomposition satisfying the given commutation relations like \eqref{equation_:commutation relations_compact}, even if it does not arise from a strict Cartan decomposition.
Exploring the necessary and sufficient conditions for such decompositions in a more general setting would be a significant contribution to Lie theory.

Second, although our general framework for a simply-connected compact Lie groups is comprehensive, we have focused on its  application  to \(\operatorname{SU}(4)\) here. Applying this theory to other simply-connected compact groups relevant to multi-qubit or qudit systems, such as \(\operatorname{SU}(2^n)\) for multi-qubit gates or \(\operatorname{SU}(d)\) for qudit gates, is conceptually straightforward but could be technically complicated. As outlined in our general theory, an explicit recipe for such an application would involve:
\begin{enumerate}
\item Identifying a maximal Abelian subalgebra \(\mathfrak{a}\) and its corresponding Weyl group;
\item Determining the appropriate $p$-lattice;
\item Analyzing the action of the resulting projective affine Weyl group on \(\mathfrak{a}\) to obtain the fundamental domain for the equivalence classes.
\end{enumerate}
A systematic study of these steps for higher-dimensional systems would be a natural and valuable continuation of our work.

Third, our analysis on equivalence classes is specifically tailored to simply-connected compact Lie groups. The KAK decomposition and the classification of equivalence classes for non-simply-connected compact groups, or for non-compact reductive Lie groups, present new mathematical subtleties that remain to be studied. For compact Lie groups, the classification is conjectured to be intricately linked to the structure of their fundamental groups.  A rigorous investigation would significantly deepen the mathematical understanding of these decompositions and classifications.

Finally, extending these results to infinite-dimensional Lie groups remains largely unexplored and is expected to be of great interest to continuous-variable (CV) quantum information processing. For CV systems where the relevant symmetry groups are often infinite-dimensional (e.g.,~groups of unitary transformations on infinite-dimensional Hilbert spaces), it is not yet known under what conditions an analogous KAK-type decomposition exists. Whether a decomposition into local and entangling operations, expressed via a Cartan-like subalgebra, can be rigorously formulated in the CV setting remains a valuable open direction.

In conclusion, this work provides a rigorous mathematical foundation for a concept of great significance to modern quantum computing. By clarifying the underlying theory, distinguishing between different notions of equivalence classes, and explicitly deriving some key results for \(\operatorname{SU}(4)\), we hope to have formulated a general technique that is both mathematically sound and practically useful, fostering a deeper connection between pure mathematics and its applications in quantum information science.

\section*{Acknowledgments}
We thank Jianxin Chen, Chang Huang, and Changjian Su for valuable discussion and feedback. DD is funded by the Shanghai Institute of Mathematics and Interdisciplinary Sciences under grant number [SIMIS-ID-2025-QT]. DD would like to thank God for all of His provisions.
YL and ZWL are supported in part by NSFC under Grant No.~12475023, Dushi Program, and a startup funding from YMSC.

\bibliographystyle{alpha}
\bibliography{ref_KAK}

\appendix

\section{Definitions of the Cartan Decomposition}\label{appendix_:Proof of the equivalence of two definitions-Cartan decomposition}
In this appendix, we prove that Definition~\ref{definition_:Cartan decomposition-Helgason} and Definition~\ref{definition_:Cartan decomposition-Knapp96} are equivalent.
\begin{proof}[Proof]

\textbf{Proof of Definition~\ref{definition_:Cartan decomposition-Helgason} $\implies$ Definition~\ref{definition_:Cartan decomposition-Knapp96}: }

Assume a real semisimple Lie algebra $\mathfrak{g}_0 = \mathfrak{k}_0 \oplus \mathfrak{p}_0$ satisfies Definition~\ref{definition_:Cartan decomposition-Helgason}.
That is, there exists a compact real form $\mathfrak{g}_k$ of $\mathfrak{g}$ such that
\[
\sigma \cdot \mathfrak{g}_k \subset \mathfrak{g}_k,
\quad
\mathfrak{k}_0=\mathfrak{g}_0 \cap \mathfrak{g}_k,
\quad
\mathfrak{p}_0=\mathfrak{g}_0 \cap\left(i \mathfrak{g}_k\right).
\]

For $X, Y \in \mathfrak{k}_0 \subset \mathfrak{g}_k$, we have $[X, Y] \in \mathfrak{g}_k$ (since $\mathfrak{g}_k$ is a subalgebra). Also $[X, Y] \in \mathfrak{g}_0$ because $\mathfrak{g}_0$ is a subalgebra. Hence
$$
[X, Y] \in \mathfrak{g}_0 \cap \mathfrak{g}_k = \mathfrak{k}_0.
$$
So $[\mathfrak{k}_0, \mathfrak{k}_0] \subset \mathfrak{k}_0$.
\\
\noindent
For $X \in \mathfrak{k}_0$ and $Y \in \mathfrak{p}_0$, write $Y = iZ$ with $Z \in \mathfrak{g}_k$. Then
$$
[X, Y] = [X, iZ] = i[X, Z] \in i\mathfrak{g}_k.
$$
Also $[X, Y] \in \mathfrak{g}_0$. Hence
$$
[X, Y] \in \mathfrak{g}_0 \cap i\mathfrak{g}_k = \mathfrak{p}_0.
$$
So $[\mathfrak{k}_0, \mathfrak{p}_0] \subset \mathfrak{p}_0$.
\\
\noindent
For $X, Y \in \mathfrak{p}_0$, write $X = iZ$, $Y = iW$ with $Z, W \in \mathfrak{g}_k$. Then
$$
[X, Y] = [iZ, iW] = -[Z, W] \in \mathfrak{g}_k.
$$
Also $[X, Y] \in \mathfrak{g}_0$. Hence
$$
[X, Y] \in \mathfrak{g}_0 \cap \mathfrak{g}_k = \mathfrak{k}_0.
$$
So $[\mathfrak{p}_0, \mathfrak{p}_0] \subset \mathfrak{k}_0$.
Therefore, the commutation relations \eqref{equation_:commutation_relations_in_Cartan_decomposition} are satisfied.

Since $\mathfrak{g}_k$ is a compact real form, its Killing form $B_{\mathfrak{g}_k}$ is negative definite. For any non-zero $X \in \mathfrak{k}_0$ and  $Y\in \mathfrak{p}_0, Z = iY \in i \mathfrak{p}_0$, we have
\[
\operatorname{ad}_X\circ \operatorname{ad}_X (iY)  =  i \operatorname{ad}_X\circ \operatorname{ad}_X (Y),
\]
hence for any non-zero $X \in \mathfrak{k}_0 \subset \mathfrak{g}_k$, we have
$$
B_{\mathfrak{g}_0}(X, X) = B_{\mathfrak{g}_k}(X, X) < 0.
$$
Hence $B_{\mathfrak{g}_0}$ is negative definite on $\mathfrak{k}_0$.
And for any non-zero $Y \in \mathfrak{p}_0$, write $Y = iZ$ with $Z \in \mathfrak{g}_k$. Then
$$
B_{\mathfrak{g}_0}(Y, Y) = B_{\mathfrak{g}}(iZ, iZ) = i^2 B_{\mathfrak{g}}(Z, Z) = -B_{\mathfrak{g}}(Z, Z) > 0,
$$
since $B_{\mathfrak{g}}(Z, Z) = B_{\mathfrak{g}_k}(Z, Z) < 0$.
Hence $B_{\mathfrak{g}_0}$ is positive definite on $\mathfrak{p}_0$.

Therefore, the Killing form definite conditions are satisfied. Therefore, Definition~\ref{definition_:Cartan decomposition-Helgason} $\implies$ Definition~\ref{definition_:Cartan decomposition-Knapp96}.

\textbf{Proof of  Definition~\ref{definition_:Cartan decomposition-Knapp96} $\implies$ Definition~\ref{definition_:Cartan decomposition-Helgason}: }

Assume $\mathfrak{g}_0 = \mathfrak{k}_0 \oplus \mathfrak{p}_0$ satisfies the conditions in Definition~\ref{definition_:Cartan decomposition-Knapp96}.

We first construct an involution. Define a linear map $\theta: \mathfrak{g}_0 \to \mathfrak{g}_0$ by
$$
\theta(X) = X \quad \text{for } X \in \mathfrak{k}_0,\qquad
\theta(Y) = -Y \quad \text{for } Y \in \mathfrak{p}_0.
$$
The commutation relations \eqref{equation_:commutation_relations_in_Cartan_decomposition} ensure that $\theta$ is a Lie algebra automorphism. Clearly $\theta^2 = \mathrm{Id}$, so $\theta$ is an involution.

Let $\mathfrak{g} = \mathfrak{g}_0 \oplus i\mathfrak{g}_0$ be the complexification of $\mathfrak{g}_0$.
Extend $\theta$ to a complex-linear involution on $\mathfrak{g}$, still denoted by $\theta$.
Let $\sigma$ be the conjugation of $\mathfrak{g}$ with respect to $\mathfrak{g}_0$, so $\sigma$ is anti-linear, i.e.,
$$
\sigma(X + iY) = X - iY, \quad \forall X,Y \in \mathfrak{g}_0.
$$
Consider the anti-linear involution $\sigma \theta : \mathfrak{g} \to \mathfrak{g}$.
Define
$$
\mathfrak{g}_k := \{ X \in \mathfrak{g} : \sigma \theta X = X \},
$$
which  is the set of fixed points of an anti-linear involution and equals $\mathfrak{g}_k = \mathfrak{k}_0 \oplus i\mathfrak{p}_0$, hence a real form of $\mathfrak{g}$.

Next we show that $\mathfrak{g}_k$ is compact.
For any non-zero $X \in \mathfrak{g}_k$, let $X = Y + iZ$ with $Y\in \mathfrak{k}_0, Z\in \mathfrak{p}_0$, then
\[
B_{\mathfrak{g}_k} (X,X) = B_{\mathfrak{g}_k} (Y + iZ, Y + iZ )
=  B_{\mathfrak{g}_k} (Y, Y) -  B_{\mathfrak{g}_k} (Z, Z).
\]
The crossing terms vanish because of commutation relations.
Since we can identify $B_{\mathfrak{g}_k}$ with $B_{\mathfrak{g}_0}$, so by definition of $B_{\mathfrak{g}_0}$, we have $B_{\mathfrak{g}_k} (Y, Y) < 0,  -  B_{\mathfrak{g}_k} (Z, Z) < 0$ unless $Y=0$ or $Z = 0$. Therefore, $B_{\mathfrak{g}_k} (X,X)< 0 $ unless $X = 0$. Hence $\mathfrak{g}_k$ is a compact real form of $\mathfrak{g}$.

By construction,  $\mathfrak{g}_0 \cap \mathfrak{g}_k$ consists of elements in $\mathfrak{g}_0$ fixed by $\theta$, which is precisely $\mathfrak{k}_0$.
And $\mathfrak{g}_0 \cap i\mathfrak{g}_k$ consists of elements in $\mathfrak{g}_0$ on which $\theta$ acts as $-1$, which is precisely $\mathfrak{p}_0$.
Also, by the definition of  $\mathfrak{g}_k$, let $X\in \mathfrak{g}_k$, then $ \sigma \theta X = X$. So  $ \sigma \theta \sigma X = \sigma \theta \theta X = \sigma X$, hence  $\sigma(\mathfrak{g}_k) \subset \mathfrak{g}_k$.

Therefore, the compact real form $\mathfrak{g}_k$ satisfies all conditions. Therefore,
Definition~\ref{definition_:Cartan decomposition-Knapp96} $\implies$ Definition~\ref{definition_:Cartan decomposition-Helgason}.
\qedhere
\end{proof}

\section{Some key results in Lie theory}\label{appendix_:Maximal Torus Theorem Some key results in Lie theory}

The Maximal Torus Theorem is a fundamental result in the theory of compact Lie groups. It states that

\begin{theorem}[Maximal Torus Theorem]\label{theorem_:Maximal Torus Theorem}
Let $G$ be a compact connected Lie group. Then the followings hold.
\begin{enumerate}
\item
\textbf{Existence of a maximal torus}: $G$ contains a maximal torus $T$, which is a maximal connected Abelian subgroup isomorphic to $(S^1)^k$.
\item
\textbf{Conjugacy of maximal tori}: Any two maximal tori in $G$ are conjugate.
\item
\textbf{Surjectivity of the exponential map}: Every element of $G$ is conjugate to an element of $T$, i.e.,
$$
G = \bigcup_{g \in G} g T g^{-1}.
$$
\end{enumerate}
\end{theorem}

This theorem is crucial for understanding the structure of compact Lie groups.  As an application it can be used to prove that the exponential map $\exp: \mathfrak{g} \to G$ is surjective for compact Lie groups.

\begin{proposition}[]\label{proposition_:exponential map compact onto}
Let $G$ be a compact connected Lie group, $\mathfrak{g}$ its Lie algebra, and
$$
\exp : \mathfrak{g} \to G
$$
the exponential map. Then the map $\exp$ is onto.
\end{proposition}

\begin{proof}[Proof]
For any $g \in G$, it is conjugate to some $h \in T$ by the maximal torus theorem. So $g = k h k^{-1}$ with $h \in T$ and $k \in G$. We know $h = \exp(Y)$ for some $Y \in \mathfrak{t}$, hence
$$
k h k^{-1} = k \exp(Y) k^{-1} = \exp(\mathrm{Ad}_k Y).
$$
Therefore, $g = \exp(\mathrm{Ad}_k Y)$ with $\mathrm{Ad}_k Y \in \mathfrak{g}$.
So every element of $G$ is in the image of $\exp$.
\qedhere
\end{proof}

Next proposition is related to the natural property of the exponential map with Lie group and Lie algebra homomorphisms and it is called the second fundamental result in the elementary theory of Lie groups in \cite[Sect.~10, Ch.~I]{Knapp1996Lie}.

\begin{proposition}[]\label{proposition_:natural property of the exponential map}
Let $G$ and $H$ be Lie groups with Lie algebras $\mathfrak{g}$ and $\mathfrak{h}$. If $G$ is connected and simply-connected and $H$ is connected, then for any Lie algebra homomorphism $\phi : \mathfrak{g} \to \mathfrak{h}$, there exists a unique Lie group homomorphism $\Phi : G \to H$ such that $d\Phi_e = \phi$. Moreover, the following commutative diagram commutes.
\begin{center}
\begin{tikzcd}
\mathfrak{g} \arrow[r, "\phi"] \arrow[d, "\exp_G"'] & \mathfrak{h} \arrow[d, "\exp_H"] \\
G \arrow[r, "\Phi"]                               & H
\end{tikzcd}
\end{center}
\end{proposition}

\section{Bell basis transformation}
\label{appendix_:Bell basis transformation}

In this appendix, we use the Bell basis transformation method to establish an isomorphism between the Lie groups $\operatorname{SO}(4)$ and $\operatorname{SU}(2)\otimes\operatorname{SU}(2)$. This transformation also changes the basis of the corresponding Lie algebras $\mathfrak{so}(4)$ and $\mathfrak{su}(2)\oplus\mathfrak{su}(2)$. This is useful for representing the matrix in the KAK decomposition. Note that it can be directly observed from the Dynkin diagram that these two Lie algebras are isomorphic, with $\mathfrak{su}(2)$ being of type A and $\mathfrak{so}(4)$ being of type D. However, here we provide a more explicit matrix transformation that is useful and beneficial for explicit calculations and engineering applications.

Let
\[
Q:=\frac{1}{\sqrt{2}}\left[\begin{array}{cccc}
1 & 0 & 0 & i \\
0 & i & 1 & 0 \\
0 & i & -1 & 0 \\
1 & 0 & 0 & -i
\end{array}\right]
\]
be the Bell basis transformation.

\begin{lemma}[]\label{lemma_:Bell_basis_transformation}
Bell basis transformation gives
\[
Q (\mathfrak{so}(4) ) Q^{\dagger} = \mathfrak{su}(2)\oplus \mathfrak{su}(2)
\]
and
\begin{equation}\label{equation_04_04:01-SO(4)}
Q(\operatorname{SO}(4)) Q^{\dagger}=\operatorname{SU}(2) \otimes \operatorname{SU}(2),
\end{equation}
\end{lemma}

\begin{proof}[Proof of Lemma~\ref{lemma_:Bell_basis_transformation}]

First we choose single-qubit Pauli matrices as a vector space basis for the Lie algebra $\mathfrak{su} (2)\oplus \mathfrak{su}(2)$ and show that it maps to a basis for \(\mathfrak{s o}(4)\). The Lie algebra of $\operatorname{SU}(2) \otimes \operatorname{SU}(2)$ is $\mathfrak{su} (2)\oplus \mathfrak{su}(2)$, a basis in one-qubit Pauli matrix is given by
\[
\frac{i}{2} \left\{ \sigma_x^1, \sigma_y^1, \sigma_z^1, \sigma_x^2, \sigma_y^2, \sigma_z^2 \right\}.
\]
Here \(\sigma_x, \sigma_y\), and \(\sigma_z\) are the Pauli matrices, and
\(\sigma_\alpha^1 \sigma_\beta^2=\sigma_\alpha^1 \otimes \sigma_\beta^2\).
In the matrix form, they are

$\frac{i}{2} \sigma_x^1 = \frac{i}{2} \begin{pmatrix} 0 & 0 & 1 & 0 \\ 0 & 0 & 0 & 1 \\ 1 & 0 & 0 & 0 \\ 0 & 1 & 0 & 0 \end{pmatrix}, $
$\frac{i}{2} \sigma_y^1
= \frac{i}{2} \begin{pmatrix} 0 & 0 & -i & 0 \\ 0 & 0 & 0 & -i \\ i & 0 & 0 & 0 \\ 0 & i & 0 & 0 \end{pmatrix},$
$\frac{i}{2} \sigma_z^1 = \frac{i}{2} \begin{pmatrix} 1 & 0 & 0 & 0 \\ 0 & 1 & 0 & 0 \\ 0 & 0 & -1 & 0 \\ 0 & 0 & 0 & -1 \end{pmatrix}, $

$
\frac{i}{2} \sigma_x^2
= \frac{i}{2} \begin{pmatrix} 0 & 1 & 0 & 0 \\ 1 & 0 & 0 & 0 \\ 0 & 0 & 0 & 1 \\ 0 & 0 & 1 & 0 \end{pmatrix}, $
$\frac{i}{2} \sigma_y^2  = \frac{i}{2} \begin{pmatrix} 0 & -i & 0 & 0 \\ i & 0 & 0 & 0 \\ 0 & 0 & 0 & -i \\ 0 & 0 & i & 0 \end{pmatrix}, $
$\frac{i}{2} \sigma_z^2  = \frac{i}{2} \begin{pmatrix} 1 & 0 & 0 & 0 \\ 0 & -1 & 0 & 0 \\ 0 & 0 & 1 & 0 \\ 0 & 0 & 0 & -1 \end{pmatrix}.$

Consider the following matrices
$$
A_1 =
\frac{1}{2} \begin{pmatrix}
0 & -1 & 0 & 0 \\
1 & 0 & 0 & 0 \\
0 & 0 & 0 & -1 \\
0 & 0 & 1 & 0
\end{pmatrix},
B_1  = \frac{1}{2} \begin{pmatrix}
0 & 1 & 0 & 0 \\
-1 & 0 & 0 & 0 \\
0 & 0 & 0 & -1 \\
0 & 0 & 1 & 0
\end{pmatrix}.
$$

$$
A_2  = \frac{1}{2} \begin{pmatrix}
0 & 0 & -1 & 0 \\
0 & 0 & 0 & 1 \\
1 & 0 & 0 & 0 \\
0 & -1 & 0 & 0
\end{pmatrix},
B_2 = \frac{1}{2} \begin{pmatrix}
0 & 0 & 1 & 0 \\
0 & 0 & 0 & 1 \\
-1 & 0 & 0 & 0 \\
0 & -1 & 0 & 0
\end{pmatrix}.
$$

$$
A_3  = \frac{1}{2} \begin{pmatrix}
0 & 0 & 0 & -1 \\
0 & 0 & -1 & 0 \\
0 & 1 & 0 & 0 \\
1 & 0 & 0 & 0
\end{pmatrix},
B_3 = \frac{1}{2} \begin{pmatrix}
0 & 0 & 0 & 1 \\
0 & 0 & -1 & 0 \\
0 & 1 & 0 & 0 \\
-1 & 0 & 0 & 0
\end{pmatrix}.
$$
One can verify that these $A_i$ and $B_i$ satisfy the $\mathfrak{su}(2)$ commutation relations, i.e.,
\begin{equation}\label{equation_04_04:AB commutation relation}
[A_i, A_j] = \epsilon_{ijk} A_k, \quad [B_i, B_j] = \epsilon_{ijk} B_k, \quad [A_i, B_j] = 0
\end{equation}
Applying $Q (\bullet) Q^{\dagger}$ to all of  $A_i$ and $B_i$  gives
\[
\begin{aligned}
Q A_1 Q^\dagger
&= \frac{i}{2} \begin{pmatrix}
0 & 1 & 0 & 0 \\
1 & 0 & 0 & 0 \\
0 & 0 & 0 & 1 \\
0 & 0 & 1 & 0
\end{pmatrix},
& \quad & &
Q B_1 Q^\dagger
&= \frac{i}{2} \begin{pmatrix}
0 & 0 & -1 & 0 \\
0 & 0 & 0 & -1 \\
-1 & 0 & 0 & 0 \\
0 & -1 & 0 & 0
\end{pmatrix},
\\
Q A_2 Q^\dagger
& = \frac{i}{2} \begin{pmatrix}
0 & i & 0 & 0 \\
-i & 0 & 0 & 0 \\
0 & 0 & 0 & i \\
0 & 0 & -i & 0
\end{pmatrix},
& \quad & &
Q B_2 Q^\dagger
& = \frac{i}{2} \begin{pmatrix}
0 & 0 & i & 0 \\
0 & 0 & 0 & i \\
-i & 0 & 0 & 0 \\
0 & -i & 0 & 0
\end{pmatrix},
\\
Q A_3 Q^\dagger
&= \frac{i}{2} \begin{pmatrix}
1 & 0 & 0 & 0 \\
0 & -1 & 0 & 0 \\
0 & 0 & 1 & 0 \\
0 & 0 & 0 & -1
\end{pmatrix},
& \quad & &
Q B_3 Q^\dagger
&= \frac{i}{2} \begin{pmatrix}
-1 & 0 & 0 & 0 \\
0 & -1 & 0 & 0 \\
0 & 0 & 1 & 0 \\
0 & 0 & 0 & 1
\end{pmatrix}.
\end{aligned}
\]
Then we have established the 1-1 correspondence between $Q (\mathfrak{so}(4)) Q^\dagger$ and $\mathfrak{su}(2) \oplus \mathfrak{su}(2)$:
\[
\begin{aligned}
Q A_1 Q^\dagger  \mapsto  & \frac{i}{2} \sigma_x^2 , & \quad & Q B_1 Q^\dagger  \mapsto  &   -\frac{i}{2} \sigma_x^1, \\
Q A_2 Q^\dagger  \mapsto  &  -  \frac{i}{2} \sigma_y^2, & \quad &  Q B_2 Q^\dagger  \mapsto  &   -\frac{i}{2} \sigma_y^1,   \\
Q A_3 Q^\dagger  \mapsto  &  \frac{i}{2} \sigma_z^2,  & \quad &   Q B_3 Q^\dagger  \mapsto  &  - \frac{i}{2} \sigma_z^1.  \\
\end{aligned}
\]

Moreover for the Lie algebra $\frac{i}{2} \left\{ \sigma_x^1, \sigma_y^1, \sigma_z^1, \sigma_x^2, \sigma_y^2, \sigma_z^2 \right\}$, we have
\[
\left[ \frac{i}{2} \sigma_\alpha^1, \frac{i}{2} \sigma_\beta^1 \right] = -\frac{1}{4} [\sigma_\alpha, \sigma_\beta] \otimes I_2 = -\frac{1}{4} (2i \epsilon_{\alpha \beta \gamma} \sigma_\gamma) \otimes I_2 = -\frac{i}{2} \epsilon_{\alpha \beta \gamma} \sigma_\gamma^1,
\]
and
\[
\left[ \frac{i}{2} \sigma_\alpha^2, \frac{i}{2} \sigma_\beta^2 \right] = -\frac{i}{2} \epsilon_{\alpha \beta \gamma} \sigma_\gamma^2,
\]
and
\[
\left[ \frac{i}{2} \sigma_\alpha^1, \frac{i}{2} \sigma_\beta^2 \right] = 0.
\]
So relating the above relations and the relations \eqref{equation_04_04:AB commutation relation} and by counting dimensions, we have showed that
\[
\begin{aligned}
Q (\bullet) Q^{\dagger}: \mathfrak{so}(4)  & \longrightarrow \mathfrak{su}(2)\oplus \mathfrak{su}(2) \\
P & \mapsto  Q P Q^{\dagger}
\end{aligned}
\]
is a Lie algebra isomorphism. This induces the isomorphism between two simply-connected Lie groups
\[
\operatorname{Spin}(4) \cong \operatorname{SU}(2)\times \operatorname{SU}(2)
\]
 which are images of $\mathfrak{so}(4)$ and $\mathfrak{su}(2)\oplus \mathfrak{su}(2)$ under the exponential map respectively.

Moreover, since
\[
\operatorname{SO}(4) \cong \operatorname{Spin}(4)/\{I, -I\}
\]
and
\[
\operatorname{SU}(2)\otimes \operatorname{SU}(2) \cong \operatorname{SU}(2)\times \operatorname{SU}(2) /\{(I,I), (-I, -I)\},
\]
and
\[
Q^{\dagger} (\exp(i \pi \sigma_x^{2} - i \pi \sigma_x^{1}) ) Q
= \exp (2\pi A_1 + 2\pi B_1)
= I_4,
\]
 which implies that  the element $(-I, -I) = (\exp(- i \pi \sigma_x^{1}), \exp (i \pi \sigma_x^{2} ) )$ in $\operatorname{SU}(2)\times \operatorname{SU}(2)$ corresponds to $I_4$ in $\operatorname{SO}(4)$ and $-I_4$ in $\operatorname{Spin}(4)$,
thus
\[
\operatorname{SO}(4)  \cong \operatorname{SU}(2)\otimes \operatorname{SU}(2).
\]
\qedhere
\end{proof}

\section{Proof of commutation relations of $\mathfrak{k}$ and $\mathfrak{p}$ in $\mathfrak{su}(4)$}\label{appendix_:Proof of Commutation relations}

Given a vector space decomposition $\mathfrak{su}(4) = \mathfrak{k} \oplus \mathfrak{p}$ where
\[
\begin{aligned}
\mathfrak{k}
= &
i \left\{ \sigma_x^1, \sigma_y^1, \sigma_z^1, \sigma_x^2, \sigma_y^2, \sigma_z^2 \right\}, \\
\mathfrak{p}
= &
i \left\{ \sigma_x^1 \sigma_x^2, \sigma_x^1 \sigma_y^2, \sigma_x^1 \sigma_z^2, \sigma_y^1 \sigma_x^2, \sigma_y^1 \sigma_y^2, \sigma_y^1 \sigma_z^2, \sigma_z^1 \sigma_x^2, \sigma_z^1 \sigma_y^2, \sigma_z^1 \sigma_z^2 \right\},
\end{aligned}
\]
we show that $\mathfrak{k}, \mathfrak{p}$ satisfy the commutation relations
\[ [\mathfrak{k}, \mathfrak{k}] \subset \mathfrak{k}, \quad[\mathfrak{p}, \mathfrak{k}] \subset \mathfrak{p}, \quad[\mathfrak{p}, \mathfrak{p}] \subset \mathfrak{k}.  \]

\begin{proof}
For Pauli matrices, we have
$$
[\sigma_\alpha, \sigma_\beta] = 2i \epsilon_{\alpha \beta \gamma} \sigma_\gamma,
$$
$$
\{\sigma_\alpha, \sigma_\beta\} = 2 \delta_{\alpha \beta} I,
$$
\[
\begin{aligned}
[\sigma_\alpha \otimes \sigma_\beta, \sigma_\gamma \otimes \sigma_\delta]
= &
\sigma_\alpha \sigma_\gamma  \otimes [\sigma_\beta, \sigma_\delta]
+
[\sigma_\alpha, \sigma_\gamma] \otimes \sigma_\delta \sigma_\beta
\\
= &
\sigma_\alpha \sigma_\gamma  \otimes \{\sigma_\beta, \sigma_\delta\}
-
\{\sigma_\alpha, \sigma_\gamma\} \otimes \sigma_\delta \sigma_\beta.
\end{aligned}
\]
So
\[
\left[ \frac{i}{2} \sigma_\alpha^1, \frac{i}{2} \sigma_\beta^1 \right] = -\frac{1}{4} [\sigma_\alpha, \sigma_\beta] \otimes I_2 = -\frac{1}{4} (2i \epsilon_{\alpha \beta \gamma} \sigma_\gamma) \otimes I_2 = -\frac{i}{2} \epsilon_{\alpha \beta \gamma} \sigma_\gamma^1 \in \mathfrak{k}.
\]
and
\[
\left[ \frac{i}{2} \sigma_\alpha^2, \frac{i}{2} \sigma_\beta^2 \right] = -\frac{i}{2} \epsilon_{\alpha \beta \gamma} \sigma_\gamma^2 \in \mathfrak{k}.
\]
And
\[
\left[ \frac{i}{2} \sigma_\alpha^1, \frac{i}{2} \sigma_\beta^2 \right] = 0.
\]
So
\[
[\mathfrak{k}, \mathfrak{k} ] \subset \mathfrak{k}.
\]

The generator of $\mathfrak{p}$ is $\frac{i}{2} \sigma_\alpha^1 \sigma_\beta^2$, so
$$
\left[ \frac{i}{2} \sigma_\gamma^1, \frac{i}{2} \sigma_\alpha^1 \sigma_\beta^2 \right] = -\frac{1}{4} [\sigma_\gamma, \sigma_\alpha] \otimes \sigma_\beta = -\frac{1}{4} (2i \epsilon_{\gamma \alpha \delta} \sigma_\delta) \otimes \sigma_\beta = -\frac{i}{2} \epsilon_{\gamma \alpha \delta} \sigma_\delta^1 \sigma_\beta^2 \in \mathfrak{p},
$$
and
$$
\left[ \frac{i}{2} \sigma_\gamma^2, \frac{i}{2} \sigma_\alpha^1 \sigma_\beta^2 \right] = -\frac{1}{4} \sigma_\alpha \otimes [\sigma_\gamma, \sigma_\beta] = -\frac{i}{2} \epsilon_{\gamma \beta \delta} \sigma_\alpha^1 \sigma_\delta^2 \in \mathfrak{p}.
$$
So
\[
[\mathfrak{k}, \mathfrak{p} ] \subset \mathfrak{p}.
\]

For two generators $\frac{i}{2} \sigma_\alpha^1 \sigma_\beta^2$ and  $\frac{i}{2} \sigma_\gamma^1 \sigma_\delta^2$ in $\mathfrak{p}$, we have
$$
\begin{aligned}
\left[
\frac{i}{2} \sigma_\alpha^1 \sigma_\beta^2,
\frac{i}{2} \sigma_\gamma^1 \sigma_\delta^2
\right]
= &
-\frac{1}{4} [\sigma_\alpha^1 \sigma_\beta^2, \sigma_\gamma^1 \sigma_\delta^2]
\\
= &
-\frac{1}{4} [\sigma_\alpha \otimes \sigma_\beta, \sigma_\gamma \otimes \sigma_\delta]
\\
=  &
-\frac{1}{4}
\left( \sigma_\alpha \sigma_\gamma  \otimes \{\sigma_\beta, \sigma_\delta\}
-
\{\sigma_\alpha, \sigma_\gamma\} \otimes \sigma_\delta \sigma_\beta \right)
\\
= &
-\frac{1}{4} \left( (\delta_{\alpha \gamma} I + i \epsilon_{\alpha \gamma \mu} \sigma_\mu) \otimes (2 \delta_{\beta \delta} I) - (2 \delta_{\alpha \gamma} I) \otimes (\delta_{\beta \delta} I + i \epsilon_{\beta \delta \nu} \sigma_\nu) \right)
\\
= &
-\frac{1}{4} \left( 2 \delta_{\alpha \gamma} \delta_{\beta \delta} I \otimes I + 2i \epsilon_{\alpha \gamma \mu} \delta_{\beta \delta} \sigma_\mu \otimes I - 2 \delta_{\alpha \gamma} \delta_{\beta \delta} I \otimes I - 2i \delta_{\alpha \gamma} \epsilon_{\beta \delta \nu} I \otimes \sigma_\nu \right)
\\
= &
-\frac{i}{2} \left( \epsilon_{\alpha \gamma \mu} \delta_{\beta \delta} \sigma_\mu^1 - \delta_{\alpha \gamma} \epsilon_{\beta \delta \nu} \sigma_\nu^2 \right) \in \mathfrak{k}.
\end{aligned}
$$
Hence,
\[
[\mathfrak{p}, \mathfrak{p} ]\subset \mathfrak{k}.
\]
\qedhere
\end{proof}

\section{Calculation of the root system}\label{appendix_:Calculation of the root system}
Here we give a detailed proof of  \eqref{equation_:root system of su_4}.
Recall that $X_{i}$ are elements from $\mathfrak{k}$ and $\mathfrak{p}$,
\[
\begin{aligned}
&X_1 = i \sigma_x^1 , & & X_2 = i \sigma_y^1 , & & X_3 = i \sigma_z^1 , \\
&X_4 = i \sigma_x^2 , & & X_5 = i \sigma_y^2 , &&X_6 = i \sigma_z^2 , \\
&X_7 = i \sigma_x^1 \sigma_x^2 ,& &X_8 = i \sigma_x^1 \sigma_y^2 , &&X_9 = i \sigma_x^1 \sigma_z^2 , \\
&X_{10} = i \sigma_y^1 \sigma_x^2 ,& &X_{11} = i \sigma_y^1 \sigma_y^2 , & &X_{12} = i \sigma_y^1 \sigma_z^2 , \\
&X_{13} = i \sigma_z^1 \sigma_x^2 , & &X_{14} = i \sigma_z^1 \sigma_y^2 , & &X_{15} = i \sigma_z^1 \sigma_z^2 .\\
\end{aligned}
\]
We already know that $\mathfrak{a}^{\mathbb{C}} = \langle X_7, X_{11}, X_{15}\rangle_{\mathbb{C}}$ forms a Cartan subalgebra for $\mathfrak{g}^{\mathbb{C}}$. Then we can get  the following table of the Lie bracket calculations, where $X_j$ are the first column elements and $X_k$ are the first row elements.

\begin{center}
\begin{tabular}{l|lllllll|lllllll}
\(\left[X_j, X_k\right]/2\) & \(X_1\) & \(X_2\) & \(X_3\) & \(X_4\) & \(X_5\) & \(X_6\) &  & \(X_8\) & \(X_9\) & \(X_{10}\) &  & \(X_{12}\) & \(X_{13}\) & \(X_{14}\)
\\
\hline
\(X_7\) & 0 & \(-X_{13}\) & \(X_{10}\) & 0 & \(-X_9\) & \(X_8\) &  & \(-X_6\) & \(X_5\) & \(-X_3\) &  & 0 & \(X_2\) & 0
\\
\(X_{11}\) & \(X_{14}\) & 0 & \(-X_8\) & \(X_{12}\) & 0 & \(-X_{10}\) &  & \(X_3\) & 0 & \(X_6\) &  & \(-X_4\) & 0 & \(-X_1\)
\\
\(X_{15}\) & \(-X_{12}\) & \(X_9\) & 0 & \(-X_{14}\) & \(X_{13}\) & 0 &  & 0 & \(-X_2\) & 0 &  & \(X_1\) & \(-X_5\) & \(X_4\)
\end{tabular}
\end{center}

Let $E_{(a,b,c)}$ ($a,b,c\in \mathbb{C}$) be the vector in $\mathfrak{g}^{\mathbb{C}}$ such that
\[
\left\{
\begin{aligned}
[X_7, E_{(a,b,c)}] = a\cdot E_{(a,b,c)}, \\
[X_{11}, E_{(a,b,c)}] = b\cdot E_{(a,b,c)}, \\
[X_{15}, E_{(a,b,c)}] = c\cdot E_{(a,b,c)}. \\
\end{aligned}
\right.
\]
Then we have
\[
\begin{aligned}
E_{(0,2i,2i)} &= X_1 - i X_{14} - X_4 + i X_{12}, \\
E_{(0,2i,-2i)} &= X_1 - i X_{14} + X_4 - i X_{12}, \\
E_{(0,-2i,2i)} &= X_1 + i X_{14} + X_4 + i X_{12}, \\
E_{(0,-2i,-2i)} &= X_1 + i X_{14} - X_4 - i X_{12}, \\
E_{(2i,0,2i)} &= X_2 + i X_{13} - X_5 - i X_{9}, \\
E_{(2i,0,-2i)} &= X_2 + i X_{13} + X_5 + i X_{9}, \\
E_{(-2i,0,2i)} &= X_2 - i X_{13} + X_5 - i X_{9}, \\
E_{(-2i,0,-2i)} &= X_2 - i X_{13} - X_5 + i X_{9}, \\
E_{(2i,2i,0)} &= X_3 - i X_{10} - X_6 + i X_{8}, \\
E_{(2i,-2i,0)} &= X_3 - i X_{10} + X_6 - i X_{8}, \\
E_{(-2i,2i,0)} &= X_3 + i X_{10} + X_6 + i X_{8}, \\
E_{(-2i,-2i,0)} &= X_3 + i X_{10} - X_6 - i X_{8}. \\
\end{aligned}
\]
This gives the root system \eqref{equation_:root system of su_4} of the complexification of $\mathfrak{su}(4)$ with respect to the Cartan subalgebra $\mathfrak{a}^{\mathbb{C}}$. For example, $E_{(0,2i,2i)} = X_1 - i X_{14} - X_4 + i X_{12}$ spans the space $\mathfrak{g}^{\mathbb{C}}_{e_{y}+ e_{z}}$, $E_{(0,2i,-2i)} = X_1 - i X_{14} + X_4 - i X_{12}, $ spans the space $\mathfrak{g}^{\mathbb{C}}_{e_{y}- e_{z}}$, etc.

\section{A more direct proof of uniqueness}\label{appendix_:more_direct_proof_of_uniqueness}
In this appendix we will give a more direct proof of uniqueness in Proposition~\ref{proposition_:criterion_equivalence_classes_up_to_global_phases} only using the criterion in Proposition~\ref{proposition_:equivalence_relation_global_phases}.
\begin{proof}
As for uniqueness, suppose
\[
\vec{c} = [c_1, c_2, c_3], \quad \vec{c'} = [c_1^{'}, c_2^{'}, c_3^{'}],
\]
are two elements in the set $P$ and represent the same equivalence class. Then by definition we have
\begin{equation}\label{equation_:equivalence class_c}
\pi/2 > c_1 \geq c_2 \geq c_3 \geq 0,
c_1+c_2 \leq \pi/2,
\text { if } c_3=0, \text { then } c_1 \leq \pi/4,
\end{equation}
and
\begin{equation}\label{equation_:equivalence class_c_prime}
\pi/2> c_1^{'} \geq c_2^{'} \geq c_3^{'} \geq 0,
c_1^{'}+c_2^{'} \leq \pi/2,
\text { if } c_3^{'}=0, \text { then } c_1^{'} \leq \pi/4.
\end{equation}
Since two elements $\vec{c} = [c_1, c_2, c_3], \vec{c'} = [c_1^{'}, c_2^{'}, c_3^{'}]$ are in the same orbit under the group action, then there exists $g\in \Gamma_{p}$ such that $ \vec{c'}  = g \cdot \vec{c}$.
By the semi-direct product structure of the projective affine Weyl group
\[
\Gamma_p = \mathfrak{a}_p  \rtimes W(U,K),
\]
since $\mathfrak{a}_p$ is a normal subgroup of $\Gamma_{p}$, so any element $g\in \Gamma_p$ can be represented as a product $g = r \cdot s $ for some $r \in \mathfrak{a}_p$ and $s\in W(U,K)$.
Therefore, $ \vec{c'}  = r \cdot s  \cdot \vec{c} $. By the form of elements in the Weyl group (cf. Corollary~\ref{corollary_:Weyl group structure}), we only need to consider the following two cases (Case A and Case B).

\textbf{Case A: } If $s$ is a permutation of entries of $\vec{c}$, let $s\cdot \vec{c}$ be the vector after the action of $s$, then
\[
\frac{\pi}{2} > (s\cdot \vec{c})_1, (s\cdot \vec{c})_2, (s\cdot \vec{c})_3 \geq 0.
\]
Thus the only possible translation $r$ is the identity $0$. Thus $\vec{c'} = s\cdot \vec{c}$. Since $ [c_1, c_2, c_3], [c_1^{'}, c_2^{'}, c_3^{'}]$ satisfy the descending order, then $s$ is the identity, so $\vec{c'}  = \vec{c}$.

\textbf{Case B: } If $s$ is a permutation with sign flips of two entries, we consider the following four cases (Case 1-4):

\textbf{Case 1: } if $c_1 = c_2 = c_3 = 0 $ , then for any $s\in W(U,K)$, $s\cdot \vec{c} = \vec{c}$, and the only translation to make the entries of $s \cdot \vec{c}$ in the interval $[0,\pi/2)$ is the identity $0$. Hence, in this case $\vec{c} = \vec{c'}$.

\textbf{Case 2: } if $c_1>0, c_2=c_3 = 0 $, since via a shift the last two entries $[0,0]$ should lie in the interval $[0,\pi/2)$, then $s\cdot \vec{c} $ can only be in the following form
\[
s\cdot \vec{c} = [c_1, 0, 0]  \quad \text{ or } \quad s\cdot \vec{c} = [-c_1, 0 ,0 ]
\]
to make the entries coincide.
If  $s\cdot \vec{c} = [c_1, 0, 0]$, then the translation $r$ must be $0$, thus $ \vec{c'}  = \vec{c} $. If $s\cdot \vec{c} = [-c_1, 0 ,0 ]$, then the translation $r$ must be $r\cdot s \cdot \vec{c} = [\pi/2-c_1, 0 ,0 ]$, which leads to $c_1^{'} = \pi/2 - c_1$. Since $c_1 \leq \pi/4, c_1^{'} \leq \pi/4$ by \eqref{equation_:equivalence class_c} and \eqref{equation_:equivalence class_c_prime}, so $c_1 = c_1^{'} = \pi/4$. This case leads to $\vec{c} = \vec{c'}$.

\textbf{Case 3: } if $c_1, c_2 > 0 , c_3 = 0$, then $s\cdot \vec{c} $ can only be in the following forms since via a shift the last  entry lies in the interval $[0,\pi/2)$,
\[
\begin{aligned}
[c_1, c_2, 0], [-c_1, c_2, 0], [c_1, -c_2, 0], [-c_1, -c_2, 0], \\
 [c_2, c_1, 0], [-c_2, c_1, 0], [c_2, -c_1, 0], [-c_2, -c_1, 0].
\end{aligned}
\]
To make the entries in the interval $[0,\pi/2)$, after taking a shift $r$ and considering the descending order in $\vec{c'} = [c_1^{'}, c_2^{'}, c_3^{'}]$ and $c_1 \geq c_2$ and $c_1+ c_2 \leq \pi/2, c_1^{'}+ c_2^{'} \leq \pi/2$, we can only get the following possible forms
\begin{enumerate}
\item
$ \ [c_1^{'}, c_2^{'}, c_3^{'}]  = [c_1, c_2, 0]$,
\item
$ [c_1^{'}, c_2^{'}, c_3^{'}] = [\pi/2-c_1, c_2, 0]$,
\item
$  [c_1^{'}, c_2^{'}, c_3^{'}]  = [\pi/2-c_2, c_1, 0]$,
\item
$ [c_1^{'}, c_2^{'}, c_3^{'}]  = [\pi/2-c_2, \pi/2-c_1, 0]$,
\end{enumerate}
\noindent
\textbf{3-1} gives $c_1 = c_2 = c_1^{'} = c_2^{'}$. \\
\noindent
\textbf{3-2} gives $\pi/2-c_1 = c_1^{'}$, and by \eqref{equation_:equivalence class_c} and \eqref{equation_:equivalence class_c_prime} this leads to  $c_1 =  c_1^{'} = \pi/4$. \\
\noindent
\textbf{3-3} gives $c_1^{'} + c_2 = \pi/2, c_2^{'} = c_1$, thus $c_2\geq c_2^{'} = c_1\geq c_2$, leading to $c_2=c_1 = c_2^{'}$. So $c_1= c_2 =c_1^{'}  = c_2^{'} =  \pi/4$. \\
\noindent
\textbf{3-4} gives $ c_1^{'}+ c_2 = \pi/2, c_2^{'}+c_1 =\pi/2 $, this gives $\pi \geq c_1^{'}+ c_2+ c_2^{'}+c_1=\pi $. By \eqref{equation_:equivalence class_c} and \eqref{equation_:equivalence class_c_prime}, we have $c_1+c_2 = \pi/2$ and $c_1^{'}+c_2^{'} = \pi/2$,  so $c_1 = c_1^{'}$. Since $c_2 , c_2^{'} \leq c_1 = c_1^{'} \leq \pi/4$, so it can only lead to $c_1= c_2 =c_1^{'}  = c_2^{'} =  \pi/4 $.

\textbf{Case 4: } if $ c_3 > 0 $, then the only possible cases after the action of $r\cdot s$ are
\begin{enumerate}
\item
$\{c_1^{'}, c_2^{'}, c_3^{'} \}= \{c_1, \pi /2 - c_3, \pi/2  - c_2\},  $
\item
$\{c_1^{'}, c_2^{'}, c_3^{'}\} = \{\pi /2 -  c_3, c_2,  \pi /2 -  c_1\}, $
\item
$\{c_1^{'}, c_2^{'}, c_3^{'}\} = \{\pi /2 -  c_2,  \pi /2 - c_1, c_3\}. $
\end{enumerate}
where the sets on both side are identical.
Considering the descending order and $c_1+c_2 \leq \pi/2, c_1^{'}+c_2^{'} \leq \pi/2$, the above possible cases can only be the following three types.
\begin{enumerate}
\item
$\vec{c'} = [c_1^{'}, c_2^{'}, c_3^{'}] = [\pi /2 - c_3, \pi/2  - c_2, c_1]$,
\item
$\vec{c'}= [c_1^{'}, c_2^{'}, c_3^{'}]  = [\pi /2 -  c_3,  \pi /2 -  c_1, c_2], $
\item
$\vec{c'} = [c_1^{'}, c_2^{'}, c_3^{'}] = [\pi /2 -  c_2,  \pi /2 - c_1, c_3]. $
\end{enumerate}
\noindent
\textbf{4-1} gives
\[
c_1^{'}+ c_3 = \pi/2, c_2^{'}+ c_2 = \pi/2, c_1 = c_3^{'}.
\]
So
\[
\pi = c_1^{'}+ c_3+ c_2^{'}+ c_2\leq  c_1^{'}+ c_1+ c_2^{'}+ c_2 \leq \pi,
\]
and the identity holds only if $c_1 = c_2 = c_1^{'} = c_2^{'}= \pi/4$. This leads to $c_3 = c_3^{'}= \pi/4$.
\\
\noindent
\textbf{4-2} gives
\[
c_1^{'}+ c_3 =\pi/2, c_2^{'}+ c_1 = \pi/2, c_3^{'}= c_2.
\]
So
\[
\pi = c_1^{'}+ c_3 + c_2^{'}+ c_1 \leq c_1^{'}+ c_2 + c_2^{'}+ c_1 \leq \pi,
\]
and the identity holds only if $c_1 = c_2 = c_1^{'} = c_2^{'}= \pi/4$. This leads to $c_3 = c_3^{'}= \pi/4$.
\\
\noindent
\textbf{4-3} gives
\[
c_1^{'} + c_2 =\pi/2, c_2^{'}+ c_1 =\pi/2, c_3^{'} = c_3.
\]
So
\[
\pi = c_1^{'} + c_2 +  c_2^{'}+ c_1 \leq \pi,
\]
and the identity holds only if $c_1 = c_2 = c_1^{'} = c_2^{'}= \pi/4$.
\\
\noindent
Thus all three cases  can only lead to
\[
[c_1^{'}, c_2^{'}, c_3^{'}] = [c_1, c_2, c_3].
\]

In sum, we have concluded that $ \vec{c'}  = \vec{c} $, which proves the uniqueness of each element  in the projective tetrahedral cell $P$.
\qedhere
\end{proof}

\end{document}